\documentclass[12pt]{article} 

\usepackage[utf8]{inputenc} 

\usepackage{amsmath}
\usepackage{graphicx,psfrag,epsf}
\usepackage{enumerate}
\usepackage{natbib}
\usepackage{url} 

\newcommand{\blind}{1}

\RequirePackage{amsthm,amsmath,amsfonts}
\RequirePackage[colorlinks,citecolor=blue,urlcolor=blue]{hyperref}
\usepackage{multirow}
\usepackage{algorithm2e}
\usepackage{comment}

\usepackage{pgfplots}
\usepackage{pgfplotstable}
\usepackage{here}
\usepackage{float}
\usepackage{caption}

\usepgfplotslibrary{statistics}

\numberwithin{equation}{section}
\newtheorem{thm}{Theorem}[section]
\newtheorem{prop}{Proposition}
\newtheorem{lemma}{Lemma}[section]
\newtheorem{proc}{Procedure}
\newtheorem{rem}{Remark}[section]
\newtheorem{assumption}{Assumption}

\newcommand{\kR}{\mathbb{R}}
\newcommand{\kN}{\mathbb{N}}
\newcommand{\Xt}{(X_t)_{t_0\leq t\leq t_n}}

\newcommand{\norme}[1]{\left\Vert #1\right\Vert}

\newcommand\OU{Ornstein-Uhlenbeck\xspace}
\newcommand{\BRD}{Brownian motion with drift\xspace}
\newcommand{\fBm}{fractional Brownian motion\xspace}
\newcommand{\Pv}{\textit{p}-value\xspace}
\newcommand{\Pvs}{\textit{p}-values\xspace}

\usepackage{titlesec} 

\titleformat{\subsection}[runin]
  {\normalfont\bfseries}{\thesubsection}{1em}{}

\begin{document}




%
%

\if1\blind
{
  \title{\bf A statistical analysis of particle trajectories in living cells}
  \author{V. Briane$^{\star \ddag}$, C. Kervrann$^\star$, M. Vimond$^\ddag$\\
    INRIA, Centre de Rennes Bretagne Atlantique, Serpico Project-Team$^\star$\\
    CREST (Ensai,  Universit\'e Bretagne Loire)}
  \maketitle
} \fi

\if0\blind
{
  \bigskip
  \bigskip
  \bigskip
  \begin{center}
    {\LARGE\bf A statistical analysis of particle trajectories in living cells}
\end{center}
  \medskip
} \fi

\bigskip
\begin{abstract}
Recent advances in molecular biology and fluorescence microscopy imaging have made possible the inference of the dynamics of single mole\-cules in living cells. Such inference allows to determine the organization and function of the cell. The trajectories of particles in the cells, computed with tracking algorithms, can be modelled with diffusion processes. Three types of diffusion are considered : (i) free diffusion; (ii) subdiffusion or (iii) superdiffusion.
The Mean Square Displacement (MSD) is generally used to determine the different types of dynamics of the particles in living cells \citep{qian1991single}.
We propose here a non-parametric three-decision test as an alternative to the MSD method.
The rejection of the null hypothesis -- free diffusion -- is  accompanied by claims of the direction of the alternative (subdiffusion or a superdiffusion).
We study the asymptotic behaviour of the test statistic under the null hypothesis, and under parametric alternatives which are currently considered in the biophysics literature, \citep{monnier2012bayesian} for example.
In addition, we adapt the procedure of \cite{benjamini2000adaptive} to fit with the three-decision test setting, in order to apply the test procedure to a collection of independent trajectories.
The performance of our procedure is much better than the MSD method as confirmed by Monte Carlo experiments. The method is demonstrated on real data sets corresponding to protein dynamics observed in fluorescence microscopy.
\end{abstract}

\noindent%
\vfill

\newpage
\section{Introduction}

A cell is composed of lots of structures in interaction with each other. They continuously exchange biological material, such as proteins, directly via the cytosol or via networks of polymerised filaments namely the microtubules,  actin filaments and intermediate filaments. The dynamics of these proteins determine the organization and function of the cell \citep[chapter 9]{bressloff2014stochastic}.
The traffic is known to be oriented and it is established that local dynamics of proteins obey to biophysical laws, including  subdiffusion (diffusion in a closed domain or in an open but crowded area), free diffusion (or Brownian motion) and superdiffusion (active transport along the microtubules). 
Then, inference on the modes of mobility of molecules is central in cell biology since it reflects the interaction of the  structures of the cell. For instance the postsynaptic AMPA-type glutamate receptors (AMPARs) is a protein involved in the fast excitatory synaptic transmission : it plays a crucial part in many aspects of brain functions including learning, memory and cognition. The dynamics of AMPARS determine synaptic transmission : aberrant AMPAR trafficking is implicated in neurodegenerative process, see \cite{henley2011routes}. \cite{hoze2012heterogeneity} model their motion with diffusions confined in a potential well. 
As an other example, \cite{lagache2009quantitative} model the dynamics of a virus invading a cell to infer its mean arrival time to the cell nucleus where it replicates. In the model of \cite{lagache2009quantitative}, the dynamic of the virus alternates between superdiffusion and Brownian motion. In this paper, we are interested by the classification of individual intracellular particle trajectories  into three modes of mobility: subdiffusion, free diffusion and superdiffusion (see Figure \ref{fig:typical_traj}). Usually, in the biophysics literature, the definition of these dynamics is related to the criterion of the mean square displacement (MSD), see for example \citep{qian1991single}. 
Given a particle trajectory $(X_t)_{t>0},$ the MSD is defined as the function,
\begin{equation}
\mathrm{MSD}(t)=\mathbb{E}\left(\left\|X_{t+t_0}-X_{t_0}\right\|^2\right),
\end{equation}
where $\|\cdot\|$ is the euclidean norm and $\mathbb{E}$ is the expectation of the probability space. 
If the MSD is linear ($\mathrm{MSD}(t)\propto t$), the trajectory is a free diffusion. 
In the biophysics literature \citep{qian1991single,saxton1997single}, this kind of diffusion is associated to the Brownian motion (or Wiener process in mathematics). \cite{kou2008stochastic} defines the physical Brownian motion via the Langevin equation with white noise which is different from the biophysical Brownian motion. In this case we have $\mathrm{MSD}(t)\propto t$ for large $t$ only. \cite{bressloff2014stochastic} argues that both definitions of the Brownian motion can be used to model intracellular dynamics in the case where the particle evolves freely inside the cytosol or along the plasma membrane. We decided to pick the biophysical definition corresponding to the Wiener process in mathematics as \cite{lysy2016model} did.
If the MSD is sublinear (for instance $\mathrm{MSD}(t) \sim t^\beta$ with $\beta\in(0,1)$), the trajectory is a subdiffusion \cite{lysy2016model}. Subdiffusion, which includes confined diffusion and anomalous diffusion, are the translations of several biological scenarios.
Confined or restricted diffusion \citep{metzler2000random,hoze2012heterogeneity} is characteristic of trapped particles: the particle  encounters a binding site, then it pauses for a while before dissociating and moving away. 
Anomalous diffusion includes particles which encounters  dynamic or fixed obstacles \citep{saxton1994anomalous,berry2014anomalous}, or particles  slowed by the contrary current due to the viscoelastic properties of the cytoplasm. 
In this paper, we will not distinguish confined and anomalous diffusion and consider that both are subdiffusion. \citet{meroz2015toolbox} presents a wide range of models for subdiffusion including fractional Brownian motion and the \OU process. The \OU process is widely used for modeling subdiffusion as it is the solution of the overdamped Langevin equation \citep{schuss2009theory,hoze2012heterogeneity}.
In cell biology and biophysics, superdiffusions model the motion of molecular motors and their cargo: the motion is faster and in a specific direction.
The main type of active intracellular transport involves molecular motors which carry particles (called in this context cargo) along microtubular filament tracks.
Superdiffusions are associated to the case where $\mathrm{MSD}(t) \sim t^\beta$ with $\beta>1$ \citep{feder1996constrained}.


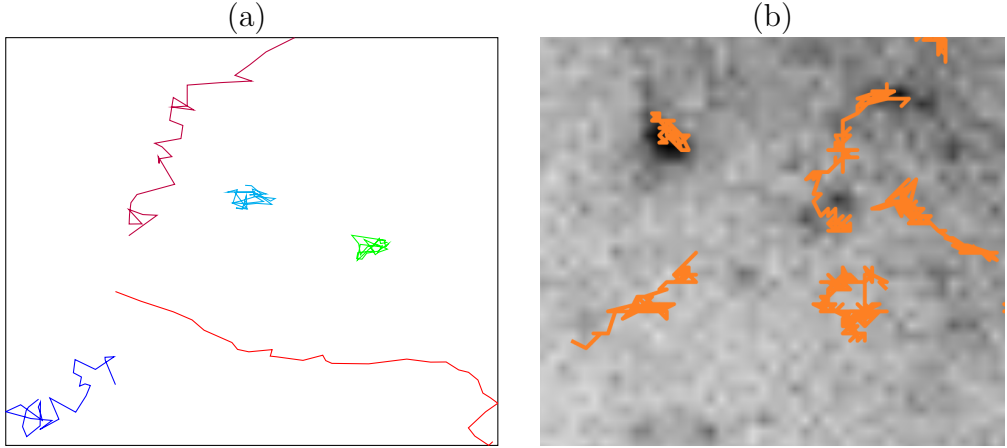
\begin{figure}[t]
\centering
\begin{tabular}{cc}
(a) & (b)\\
\pgfplotstableread {fig/typ_traj.dat} {\loadedtable}
\pgfplotsset{
    compat=newest,
    every axis/.append style={
        legend image post style={xscale=0.5}
    }
}
\begin{tikzpicture}[baseline]
\begin{axis}
[xlabel=,ylabel=,
minor tick num=4,
xtick=\empty,ytick=\empty,
enlarge x limits=false,
enlarge y limits=false,
width=0.5\linewidth,
] 

\addplot[color=blue] table[x=x_br,y=y_br] from \loadedtable;
\addplot[color=cyan] table[x=x_ou,y=y_ou] from \loadedtable;
\addplot[color=purple] table[x=x_brd,y=y_brd] from \loadedtable;
\addplot[color=red] table[x=x_fbm_sup,y=y_fbm_sup] from \loadedtable;
\addplot[color=green] table[x=x_fbm_sub,y=y_fbm_sub] from \loadedtable;


\end{axis}
\end{tikzpicture}
&
\includegraphics[width=0.38\linewidth]{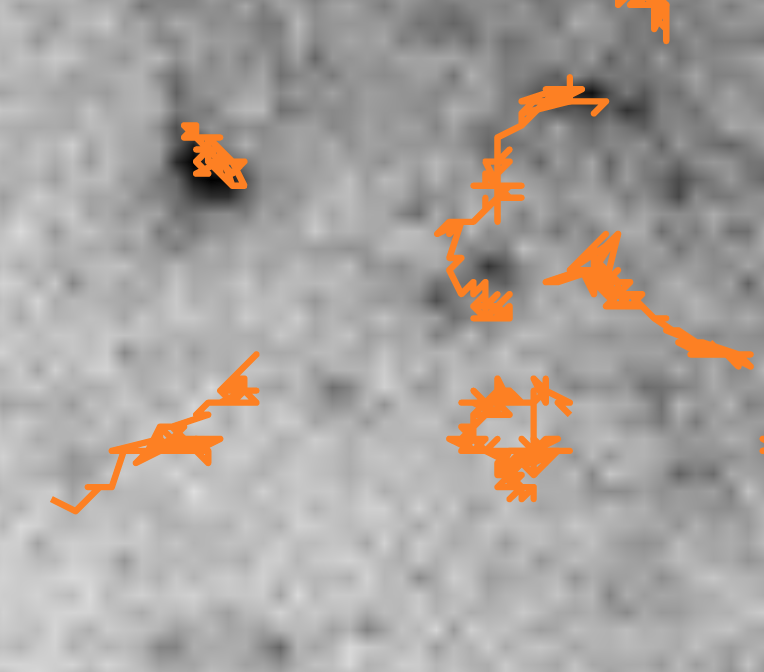}
\end{tabular}
\caption{Representative trajectories from (a) simulated data, (b) a Rab11a protein sequence in a single cell.(Courtesy of UMR 144 CNRS Institut Curie - PICT IBiSA). For the simulated data in Figure (a), the blue trajectory is Brownian, the purple one is from a \BRD, the red one from a fractional Brownian motion (parameter $\mathfrak{h}>1/2$), in cyan it is from a \OU process and in green from a fractional Brownian motion ($\mathfrak{h}<1/2$). The parameters of the processes are given in Table \eqref{tab:param_val}.}
\label{fig:typical_traj}
\end{figure}

\subsection{The problem} 
We observe the successive positions of a single particle $X_{t_0},X_{t_1},\dots,X_{t_n}$ in the real plan at equispaced times, that is $t_{i+1}-t_{i}=\Delta$. Our aim is to decide if the trajectory is a free diffusion, a subdiffusion or a superdiffusion.
A popular statistic used to determine the motion model is the pathwise Mean Square Displacement (MSD). 
It is estimated at lag $j$ by:
\begin{equation}
\widehat{\mathrm{MSD}}(j \Delta )=\frac{1}{n-j+1}\sum\limits_{k=0}^{n-j}\|X_{t_{k+j}}-X_{t_k}\|^2.
\label{eq:msd_est}
\end{equation}
The simplest rule to classify a trajectory with the MSD is based on the least-squares estimate of the slope $\beta$ of the log-log plot of the MSD \textit{versus} time \cite{feder1996constrained}.  
\cite{didier2015asymptotic} study the limiting distribution of the pathwise MSD according to the true value of $\beta.$
Nevertheless, MSD has some limitations.\\
First the MSD statistic is a summary statistic, and does not suffice to characterize the dynamics of the trajectory. \cite{gal2013particle} present several other statistics which can be associated to MSD for trajectory analysis. \cite{lund2014spattrack} propose a decision tree for selection motion model combining MSD, Bayesian information criterion and the radius of gyration. \cite{lysy2016model} present a likelihood-based inference as an alternative to MSD for the comparison between two models of subdiffusions : fractional Brownian motion  and a generalized Langevin equation. They consider a Bayesian model to estimate the parameter of the diffusion and they use the Bayes factor to compare the models.\\
Second, the variance increases with the time lag (see Figure 5 Appendix C (Supplementary Materials)): only the first few points of the MSD may be used to estimate the slope. Moreover the MSD variance is also severely affected at short time lags by dynamic localization error and motion blur. \cite{michalet2010mean} details an iterative method, known as the Optimal Least Square Fit (OLSF) for determining the optimal number of points to obtain the best fit to MSD in the presence of localization uncertainty.\\ 
In order to take account of the variance of the MSD estimate, several authors use a set of independent trajectories rather than single trajectories. These trajectories may have different lengths but are assumed to have the same kind of motion. For instance \cite{Pisarev2015} consider weighted-least-square estimate for $\beta$ by estimating the variance of pathwise MSD. Their motion model selection is then based on the modified Akaike's information criterion. \cite{monnier2012bayesian} propose a Bayesian approach to compute relative probabilities of an arbitrary set of motion models (free, confined, anomalous or directed diffusion). In general, this averaging process can lead to oversimplication and misleading conclusions about the biological process \cite{gal2013particle}. 

\subsection{Our contribution}

In this paper, we propose a measure that circumvents some limitations of the MSD and which is efficient for classifying single trajectories. Our procedure is a three-decision test procedure \citep{shaffer1980control}. The null-hypothesis is that the observed trajectory is generated from a Brownian motion and the two distinct alternatives are subdiffusion and superdiffusion. The test statistic $T_n$ is the standardized largest distance covered by the particle from its starting point.
We interpret this measure as follows: i/ if the value of $T_n$ is low, it means that the process stayed close to its initial position and the particle may be trapped in a small area or hindered by obstacles (subdiffusion); ii/ if the value of $T_n$ is high, the particle went far to its initial position and the particle may be driven by a motor in certain direction (superdiffusion). In our model, we restrict subdiffusion and superdiffusion to processes which are solution of a stochastic differential equation. However, our procedure can be extended to others types of subdiffusion in principle. 
Then, we study the asymptotic behaviour of our procedure under the null hypothesis and four parametric models illustrating superdiffusion and subdiffusion and which are commonly considered in the biophysics literature.
As stated before, we will not distinguish confined and anomalous diffusion and consider that both are subdiffusion.
The study of the behaviour of the test statistic under all existing subdiffusions process is beyond the scope of this paper.
Such refinements will be considered for a next issue.
At the end, we derive a multiple test procedure in order to apply simultaneously the test on a collection of independent trajectories which are tracked inside the same living cell.
This procedure is an adaptation of the procedure of  \cite{benjamini2000adaptive}. Then it allows to control the false discovery rate (FDR). Moreover, in case of rejection of the null hypothesis, our multiple test procedure is able to state for which alternative (subdiffusion or superdiffusion) we reject the null hypothesis. \\

The present paper is organized as follows.
In Section \ref{sec:diffusion}, we describe the inference model and provide some examples of subdiffusion and superdiffusion. 
Our testing procedure is defined in Section \ref{sec:test}.
In Section \ref{sec:multiple_test}, we derive a multiple testing procedure for a collection of trajectories.
We carry out a simulation study and illustrate the method on real data in Section \ref{sec:exp-res}. We focus on the analysis of the Rab11a GTPase protein. This protein is involved in the trafficking of molecules from the endosomes located inside the cell to the cell plasma membrane.  The data are computed from temporal sequences of TIRF microscopy images depicting the last steps of exocytosis events observed in the region very close the plasma membrane \citep{Schafer2014}. The proofs are postponed to the appendix.
\section{Diffusion models for particle trajectories}
\label{sec:diffusion}
We observe the successive positions of a single particle in a two-dimensional space at times $t_0,t_1,\dots,t_n$. We suppose that the lag time between two consecutive observations is a constant $\Delta$. The observed trajectory of the particle is,
\begin{equation*}
\mathbb{X}_n=\left({X_{t_0},X_{t_1},\dots,X_{t_{n}}}\right),
\end{equation*}
where $X_{t_i}=(X^1_{t_i},X^2_{t_i})\in \mathbb{R}^2$ is the position of the particle at time $t_i=t_0+i\Delta$, $i=0,\ldots,n$.
This discrete trajectory is generated by a stochastic process $\Xt$ with continuous path and assumed to be solution of the stochastic differential equation (SDE) :
\begin{equation}
dX^i_t=\mu_i(X^i_t)dt+\sigma dB^{\mathfrak{h},i}_t, \quad i=1,2.
\label{eq:sde}
\end{equation}
where $B^{\mathfrak{h},i}_t$ are unobserved independent 1D- fractional Brownian motions of unknown Hurst parameter $\mathfrak{h}$, $\sigma>0$ is the unknown diffusion coefficient and $(\mu_1(x_1),\mu_2(x_2)):\kR^2 \mapsto \kR^2$ is the unknown drift term.
\begin{assumption}
We assume that $\mu_i$ fulfils the linear growth hypothesis :
\begin{equation}
\exists K>0, \quad \forall x \in \kR, \quad |\mu_i(x)|\leq K (1+|x|),
\label{eq:linear_growth}
\end{equation}
and the Lipschitz condition :
\begin{equation}
\exists M>0, \quad \forall (x,y) \in \kR^2, \quad |\mu_i(x)-\mu_i(y)|\leq M \vert x-y \vert. 
\label{eq:lipschitz}
\end{equation}
\label{as:drift}
\end{assumption}
We denote by $\mathcal{L}$ the set of functions verifying Assumption \ref{as:drift}. 
Assumption \ref{as:drift} is sufficient to ensure that SDE \eqref{eq:sde} admits a strong solution  (see \citet{nualart2002regularization} for the case $0<\mathfrak{h}\leq 1/2,$ and \citet[Chapter ~3]{mishura2008stochastic} for the case $1/2<\mathfrak{h}<1$). For a given fractional Brownian motion, we say that $(X_t)$ is a strong solution of the SDE \eqref{eq:sde} if $(X_t)$ verifies \eqref{eq:sde}, has continuous paths and that, at time $t$, $X_t$ depends only on $X_{t_0}$ and on the trajectory of the fractional Brownian motion up to time $t$.
In the following, $P_{\mathfrak{h},\mu,\sigma}$ denotes the measure induced by the stochastic process $(X_t)$ solution of \eqref{eq:sde}. This measure comprises all the finite dimensional distributions of the process that is the distribution of the vectors $(X_{t_0},\dots,X_{t_n})$, $n \in \kN*$ and $t_1<,\dots,<t_n$. We also note $\mathcal{P}=\left\{P_{\mathfrak{h},\mu,\sigma} : 0<\mathfrak{h}<1,\mu \in \mathcal{L}, \sigma>0\right\}$ the set of solutions of the SDE \eqref{eq:sde}.  


\begin{rem}
In the following, we adopt the \textit{large-sample scheme} to derive asymptotic properties of our procedure, that is the inter-observation time $\Delta$ remains fixed and the number of observations $n$ tends to infinity. 
In the experimental context of microscopic sequences, $\Delta$ is the resolution of the microscopy device while $n$ is the number of frames during which we track the particle.
Other schemes exist (see \cite[Section 6.1.3]{fuchs2013inference}) as the \textit{high-frequency scheme} for which $\Delta$ tends to zero while the duration of observation is fixed.
\end{rem}

Heuristically, a SDE models the motion of a particle in a fluid submitted to a deterministic force due to the fluid and a random force due to random collisions with other particles.
That is why we model efficiently the motion of intra-cellular particles with these processes. 
In Equation \eqref{eq:sde}, the velocity of the fluid is given by the drift $\mu$ while the term $\sigma dB_t^\mathfrak{h}$ expresses the random component of the motion due to random collisions. 


\subsection{Free diffusion}
Free diffusion or Brownian motion is the most popular process for describing particle motion suspended in a liquid \citep{einstein1956investigations}. 
It suits particularly well for describing intracellular particle motion as the interior of the cell is mainly made of a fluid called the cytosol. 
Brownian motion allows dissolved macromolecules to be passively transported without any input of energy. 
In the SDE \eqref{eq:sde}, it matches with the situation where the drift $\mu_i=0$ and $\mathfrak{h}=1/2$. 

\subsection{Subdiffusion}
We present two models of subdiffusion which are solution of a stochastic differential equation \eqref{eq:sde}. We give their MSD which is, by definition of subdiffusion, sublinear. The first subdiffusion is an example of confined diffusion while the second is an anomalous diffusion. It corresponds to two distinct biological scenarios.

In the first scenario, the particle is attracted by an external force modelled by a potential well. 
We can then use the SDE \eqref{eq:sde} with a specific form for the drift : $\mu_i(x)=-\nabla U_i(x)/\gamma_i$ where $-\nabla U_i$ is the external force of the fluid and $\gamma_i$ is the frictional coefficient. 
For instance, we may consider the Ornstein-Uhlenbeck process :
\begin{equation}
dX^i_t=-\lambda_i (X^i_t-\theta_i)dt+ \sigma dB^{1/2,i}_t,\qquad i=1,2 
\label{eq:sde_OU}
\end{equation}
where $\lambda_i>0$.
Here the particle is assumed to be trapped in a single domain, the potential $U$ is uni-modal and is approximated by a polynomial of order $2$ : $U_i(x)=(1/2)k_i(x_i-\theta_i)^2$. The parameter $k_i=\lambda_i \gamma_i$ measures the strength of attraction of the potential (related to the potential depth) while $\theta=(\theta_1,\theta_2) $ is the equilibrium position of the particle. 
The \OU process is a confined diffusion according to the MSD criterion since its MSD is sublinear,
\begin{equation}
MSD(t)=\frac{2\sigma^2(1-e^{-\lambda t})}{\lambda}\leq 2\sigma^2 t, 
\end{equation} 
here it is written in the case $\lambda_i=\lambda$ for simplicity. A subdiffusion having this form of MSD is known as a confined diffusion (\cite{monnier2012bayesian} \cite{saxton1997single} \cite{Pisarev2015}).

Anomalous diffusion can occur for two main reasons. First the particle can bind to an immobile trap that can generate long jump times \citep{saxton1996binding}. 
In this situation, its motion can be modelled by a continuous time random walk \citep{metzler2000random}. 
We will not consider this model here as it is not solution of the SDE \eqref{eq:sde}. 
Secondly, the particle can be hindered by mobile or immobile obstacles as the interior environment of cells are crowded with solutes and macromolecules \citep{bressloff2013stochastic}. 
Then, a popular model is the fractional Brownian motion \citep{jeon2011vivo}. 
It corresponds to the case $0<\mathfrak{h}<1/2$ and $\mu_i=0$ in \eqref{eq:sde},
\begin{equation}\label{eq:sde_FBM}
dX^i_t=\sigma dB^{\mathfrak{h},i}_t,\qquad i=1,2. 
\end{equation} 
Its MSD is given by:
\begin{equation}
MSD(t)=2\sigma^2 t^\beta\leq 2\sigma^2 t
\label{eq:msd_fbm}
\end{equation}
with $\beta=2\mathfrak{h}<1$. A subdiffusion having this form of MSD is known as an anomalous diffusion (\cite{monnier2012bayesian} \cite{saxton1997single} \cite{Pisarev2015}).

\subsection{Superdiffusion}

At the macroscopic level, the main type of active intracellular transport involves molecular motors which carry particles (cargo) along microtubular filament tracks. 
The molecular motors and their cargo undergo superdiffusion on a network of microtubules in order to reach a specific area quickly. 
The molecular motor moves step by step along the microtubules thanks to a mechanicochemical energy transduction process. 
A single step of the molecular motor is modelled by the so-called Brownian ratchet \citep{reimann2002brownian}. 
When we observe the motion of the molecular motor along a filament on longer time-scales (several steps), its dynamic can be approximated by a Brownian motion with constant drift (also called directed Brownian) \citep[see][]{peskin1995coordinated,elston2000macroscopic}. 

The \BRD is solution of the SDE :
\begin{equation}
dX^i_t=v_idt+\sigma dB^{1/2,i}_t,\qquad i=1,2,
\label{eq:sde_dir}
\end{equation}
where $v=(v_1,v_2)\in\mathbb{R}^2$ is the constant drift parameter modelling the velocity of the molecular motor.
The MSD of the directed Brownian motion is given by:
\begin{equation}
\mathrm{MSD_D}(t)=\norme{v}^2 t^2 + 2\sigma^2 t\geq 2\sigma^2 t
\end{equation}
It is superlinear and thus define a superdiffusion. 
Superdiffusion can also be modelled by the fractional Brownian motion with Hurst parameter $1/2<\mathfrak{h}<1$. Its MSD is given by \eqref{eq:msd_fbm} as we already said. However, this time it is superlinear as $\beta=2\mathfrak{h}>1$. However, we note  that in the biophysics literature the use of the fractional Brownian motion is mainly related to subdiffusion.

\section{A statistical test procedure for a single trajectory}
\label{sec:test}
We suppose that the trajectory $\mathbb{X}_n=(X_{t_0},\dots,X_{t_n})$ is generated from some unknown diffusion process $(X_t)$ solution of the SDE \eqref{eq:sde}. Our procedure allows to test from which type of diffusion the observed trajectory is generated. 

We derive two hypothesis test procedures : one for testing $H_{0}$ "$(X_t)$ is a free diffusion" versus $H_{1}$ "$(X_t)$ is a subdiffusion", the second for testing $H_{0}$ "$(X_t)$ is a free diffusion" versus $H_{2}$ "$(X_t)$ is a superdiffusion".
Then we aggregate the two procedures to build a three-decision procedure. 

\subsection{The test statistic}
Let us consider the standardized maximal distance $T_n$ of the process from its starting point :
\begin{equation}
T_n=\frac{D_n}{\sqrt{(t_n-t_0)\hat{\sigma}^2_n}}
\label{eq:stat:test:discret}
\end{equation}
where $D_n$ is the maximal distance of the process from its starting point, 
\begin{equation}
D_n=\max_{i=1,\dots,n}\left\| X_{t_i}-X_{t_0} \right\|_2
\label{eq:stat_test1_discr}
\end{equation}
and $\hat{\sigma}_n$ is a consistent estimator of $\sigma$. 
The choice of $\hat{\sigma}$ is discussed in Section \ref{subsec:sigma}. 
If $T_n$ is low, it means the process stays close to its initial position during the period $\left[t_0,t_n\right]$ : it is likely that it is a subdiffusion. 
On contrary, if $T_n$ is large, it means the process goes away from its starting point as a superdiffusion does with high probability.
This new measure introduces an order in the diffusion processes solution of the SDE \eqref{eq:sde}. Then, it allows to classify them into the different classes of diffusion \textit{i-e} free diffusion, superdiffusion and subdiffusion. We want to build a test whose null hypothesis is that the trajectory comes from a Brownian motion, the gold standard process in biophysics. As a consequence $T_n$ must be a pivotal statistic under the hypothesis $H_0$ that is the trajectory is Brownian.
\begin{lemma}
Let $\hat{\sigma}_n$ be a consistent estimator of $\sigma$ such that the distribution of $\hat{\sigma}_n/\sigma$ does not depend on $\sigma$. If $(X_t)$ is a Brownian Motion, the distribution of $T_n$ does not depend on $\sigma$.
\label{lemma}
\end{lemma}
Let $q_n(\alpha)$ the quantile of $T_n$ of order $\alpha \in (0,1)$ when $(X_t)$ is a Brownian motion. From Lemma \ref{lemma}, $q_n(\alpha)$ does not depend on $\sigma$.

\subsection{Two hypothesis test procedures derived from the test statistic}
First we define $\phi_{1,\alpha}$ the hypotheses test associated to $H_{0}$ versus $H_{1}$ at level $\alpha \in (0,1).$ The procedure $\phi_{1,\alpha}$ is defined through its critical region,
\begin{equation}
\mathcal{R}_{1,\alpha}=\left\{T_n<q_n(\alpha) \right\},
\label{eq:critic_reg_c_1}
\end{equation}
as the following,
\begin{align*}
\phi_{1,\alpha}(\mathbb{X}_n)=
\left\{
\begin{array}{ll}
1 & \text{if } \mathbb{X}_n\in\mathcal{R}_{1,\alpha}\\
0 & \text{otherwise.}
\end{array}
\right.
\end{align*}

Then $T_n$ has probability $\alpha$ to lie in the critical region \eqref{eq:critic_reg_c_1}.
According to Lemma \ref{lemma}, the level of the test $\phi_{1,\alpha}$ is $\alpha$,
\begin{equation}
\sup_{\sigma>0}P_{1/2,0,\sigma}\left(T_n<q_n(\alpha) \right) = \alpha.
\end{equation}

In a similar way, we can perform the test $\phi_{2,\alpha}$ by replacing subdiffusion by superdiffusion in the alternative hypothesis. 
The associated critical region is :
\begin{equation}
\mathcal{R}_{2,\alpha}=\left\{T_n>q_n(1-\alpha) \right\}.
\label{eq:critic_reg_c_2}
\end{equation}

\subsection{A three-decision test procedure}
From the two tests $\phi_{1,\alpha/2}$ and $\phi_{2,\alpha/2},$ we define a new procedure $\phi$ as follows,
\begin{align}
\left\{
\begin{array}{l}
\text{we decide } H_1 \text{ if } \mathbb{X}_n\in\mathcal{R}_{1,\alpha/2},\\
\text{we decide } H_2 \text{ if } \mathbb{X}_n\in\mathcal{R}_{2,\alpha/2},\\
\text{we do not reject } H_0 \text{ otherwise.}
\end{array}
\right.
\label{three-dec-test}
\end{align}
This procedure is well defined since the intersection of the critical region $\mathcal{R}_{1,\alpha}$ and $\mathcal{R}_{2,\alpha}$ is empty.
This procedure is a three-decision test procedure and admits three kinds of errors, see Table \ref{tab:test:error:type3}.

The first kind of errors is to reject the null hypothesis $H_0$ while $H_0$ is actually true.
The probability that this error occurs is the level of the test which is defined as,
\begin{equation}
\sup_{\sigma>0}\mathbb{E}_{1/2,0,\sigma}\left(\phi_{1,\alpha}+\phi_{2,\alpha}\right) = \alpha.
\end{equation}
We only control the occurrence of this first kind of error.
Then we draw attention that acceptance of $H_0$ "$(X_t)$ is a free diffusion" does not necessarily demonstrate that $H_0$ is true. 
It only means that data does not show any evidence against the null hypothesis.
At the end, we reject this assumption in direction to one of the alternatives at level $\alpha/2.$ 

The second type of errors occurs when we do not reject the null hypothesis while one of the alternatives is true.

The last type of errors is to reject the null hypothesis in favour to a wrong alternative. In the literature of three-decision test such an error is called a Type III error, see for example \cite{rasch2012hypothesis} and references therein.


\begin{table}
\small
\centering
\caption{The three kinds of error in a three-decision test procedure.} 
\begin{tabular}{|c|ccc|}
\hline
& \multicolumn{3}{|c|}{Decision}\\
& \multirow{2}{22mm}{\centerline{Do not}\\ \centerline{Reject $H_0$}}  & \multirow{2}{22mm}{\centerline{Decide $H_1$}} & \multirow{2}{22mm}{\centerline{Decide $H_2$}}\\
Truth &  & &\\
\hline
$H_0$ True & No error & Type I & Type I\\
$H_1$ True & Type II & No error & Type III\\
$H_2$ True & Type II & Type III & No error\\
\hline
\end{tabular}
\label{tab:test:error:type3}
\end{table}

\subsection{Choosing the estimator of $\sigma$}
\label{subsec:sigma}
Ideally, we would like to find an estimator of $\sigma$ which is consistent according to the \textit{large-sample scheme} under the hypotheses $H_0$, $H_1$ and $H_2,$ and satisfies the assumption that the distribution of $\hat{\sigma}_n/\sigma$ is free of $\sigma$ under $H_0.$
However, the \textit{large-sample scheme} is not favourable to get an estimator with such properties. 
For instance, \cite{florens1989approximate} shows that the naive maximum likelihood estimator for the drift parameter has an asymptotic bias of the order of lag time $\Delta.$
Then, the \textit{high-frequency scheme} and the \textit{rapidly increasing design} turns out to be more convenient to provide consistent estimators. In fact, in the limit, these schemes correspond to the situation in which we have a continuous observation of the process on the time interval of observation.
\cite{jiang1997nonparametric} propose non parametric estimators of both the drift and the diffusion coefficient. The consistency of these estimators is proven under the {high-frequency scheme} only. 
Therefore, in this section, we discuss about the estimation of the diffusion coefficient under the \textit{large-sample asymptotic}.

The first proposition to estimate $\sigma$ may be :
\begin{equation}
\hat{\sigma}_{1,n}^2=\frac{1}{2n\Delta}\sum\limits_{j=1}^n\|X_{t_j}-X_{t_{j-1}}\|_2^2
\label{eq:sigma_standard}
\end{equation}
Even if the estimator \eqref{eq:sigma_standard} is strongly consistent under the \textit{high-frequency scheme} for every process $(X_t)$ solution of \eqref{eq:sde} \cite[Lemma~4.2, p~212]{basawa1980statistical}, Proposition \ref{prop:sigma1} tells us that it is not the case under the \textit{large-sample scheme}.

\begin{prop}\label{prop:sigma1}
\quad \newline \vspace{-0.5cm}
\begin{itemize}
\item Under $H_0,$ $\hat{\sigma}_{1,n}$ is strongly consistent and the distribution of $\hat{\sigma}_{1,n}/\sigma$ is free of $\sigma.$
\item If $(X_t)$ is an Ornstein-Uhlenbeck process \eqref{eq:sde_OU}, $\hat{\sigma}_{1,n}^2/\sigma^2$ converges in probability to ${(1-e^{-\lambda \Delta})}/({\lambda \Delta}).$ 
\item If $(X_t)$ is a \BRD \eqref{eq:sde_dir}, $\hat{\sigma}_{1,n}^2/\sigma^2$ converges almost surely to ${\Delta \|v\|_2^2}/(2\sigma^2) + 1.$
\item If $(X_t)$ is a fractional Brownian motion \eqref{eq:sde_FBM}, $\hat{\sigma}_{1,n}^2/\sigma^2$ converges almost surely to ${\Delta^{2\mathfrak{h}-1}}.$ 
\end{itemize}
\end{prop}
A proof of Proposition \ref{prop:sigma1} is given in Appendix A.2 (Supplementary Material).
Proposition \ref{prop:sigma1} states that $\hat{\sigma}_{1,n}$ is adequate to our procedure under the null hypothesis.
However $\hat{\sigma}_{1,n}$ is asymptotically biased under some alternatives. 
Notice that if $(X_t)$ is an Ornstein-Uhlenbeck process \eqref{eq:sde_OU}, then $\hat{\sigma}_{1,n}^2$ underestimates $\sigma^2$ in average since $(1-e^{-x})/x<1$ for $x>0$. Then $T_n$ might be overvalued with this estimator, increasing Type II or type III error rate in our procedure.
If $(X_t)$ is a \BRD \eqref{eq:sde_dir}, $\hat{\sigma}_{1}^2$ overestimates  $\sigma^2$ in average. 
Then $T_n$ might be overvalued with this estimator, increasing Type II or type III error rate.
Similarly, if $(X_t)$ is a fractional Brownian motion \eqref{eq:sde_FBM}, $\hat{\sigma}_{1}^2$ underestimates $\sigma^2$ if $\mathfrak{h}<1/2$, and overestimates $\sigma^2$ if $\mathfrak{h}<1/2$.

The second suggestion to estimate $\sigma$ may be based on the second order differences rather than the first order differences,
\begin{equation}
\hat{\sigma}_{2,n}^2=\frac{1}{2n\Delta}\sum\limits_{j=1}^{n-1}\|(X_{t_{j+1}}-X_{t_j})-(X_{t_j}-X_{t_{j-1}})\|_2^2.
\label{eq:sigma_2diff}
\end{equation}
As $\hat{\sigma}^2_{1,n},$ $\hat{\sigma}_{2,n}^2$ fulfils the assumption of Lemma \ref{lemma} under $H_0$.
This estimator has the advantage of decreasing the bias under some alternatives. For instance it removes the bias in the case of the \BRD.

\subsection{Approximation of the distribution of the statistic under the null hypothesis and asymptotic behaviour of our procedure}

Theorem \ref{thm:conv:H0} gives the asymptotic behaviour of our procedure under the null hypothesis. 

\begin{thm}\label{thm:conv:H0} 
Let $(X_t)$ be a Brownian Motion on $\mathbb{R}^2.$
Let $\hat{\sigma}_n$ be a consistent estimator of the diffusion parameter $\sigma$ of $(X_t).$
The test statistic $T_n$ converges in distribution to $S_0=\sup_{0\leq s \leq 1}\left\|W_s\right\|_2$ as $n \rightarrow \infty$. Here $(W_t)$ is a standard 2D Brownian motion that is the Brownian motion of variance $\mathbf{I}_2$ and initialization $W_0=(0,0)^\top$.
\end{thm}
A proof of Proposition \ref{thm:conv:H0} is given in Appendix A.1 (Supplementary Material).
The limit distribution of the test statistic under $H_0$ admits an analytical form \cite[see][Formulae.1.1.4, p.~280]{borodin1996handbook}:
\begin{equation*}
x\in(0,+\infty)\to\sum\limits_{k=1}^\infty \frac{2 e^{-j_{0,k}^2/(2x^2)}}{j_{0,k}J_1(j_{0,k})},
\end{equation*}
where $x\geq 0$, $J_\nu$ the Bessel function of order $\nu$ and $0<j_{\nu,1}<j_{\nu,2}<\dots$ the positive zeros of $J_\nu$. 
Replacing the quantiles $q_n(\alpha)$ by the quantiles of $S_0$ in our test procedure provides us a test of asymptotic level $\alpha.$

Furthermore, Proposition \ref{prop1} gives the asymptotic behaviour of the test statistic under parametric alternatives when the estimator $\hat{\sigma}_{1,n}$ is considered (see Appendix A.3 (Supplementary Material)for a proof).
More generally, as long as the estimator $\hat\sigma_n$ of the diffusion coefficient is such that $\hat\sigma_n/\sigma$ converges in probability to a positive constant whatever the dynamic of $(X_t),$ then Proposition \ref{prop1} holds.

\begin{prop}\label{prop1}
Assume that we consider the estimator \eqref{eq:sigma_standard} in our procedure \eqref{eq:stat:test:discret}.
\begin{itemize}
\item If $(X_t)$ is an Ornstein-Uhlenbeck process \eqref{eq:sde_OU}, $T_n$ converges in probability to $0.$
\item If $(X_t)$ is a fractional Brownian motion \eqref{eq:sde_FBM} with  $0<\mathfrak{h}<1/2$, $T_n$ converges in probability to $0.$
\item If $(X_t)$ is a fractional Brownian motion \eqref{eq:sde_FBM} with  $1/2<\mathfrak{h}<1,$ $T_n$ converges in probability to $+\infty.$
\item If $(X_t)$ is a \BRD \eqref{eq:sde_dir}, $T_n$ converges in probability to $+\infty.$
\end{itemize}
\end{prop}
\noindent Note that Theorem \ref{thm:conv:H0} and Proposition \ref{prop1} allow us to control the error rates of type II and type III under parametric alternatives: the associated error rates converges to $0$ with $n.$ 

However, as in practice $n$ may be small, the asymptotic approximation of the quantiles of $T_n$ may not be accurate. 
Then the level of the test is no longer $\alpha.$ 
Since we are able to draw a sample from the distribution of $T_n$ under $H_0$ (see Algorithm 1 in Appendix B (Supplementary Materials)), we propose a Monte Carlo estimate of the quantile $q_n(x)$, $0<x<1$. 
This estimate is defined as the $[xN]^{\text{th}}$ order statistic, $q^{(N)}_n(x),$ of the sample $(T^{(1)}_n,\dots,T^{(N)}_n)$. 
Table \ref{tab:quantiles} shows that there is a significant difference between asymptotic and non asymptotic quantiles. 
As expected, as $n \rightarrow \infty$, $q_{n}(\alpha)$ converges to $q(\alpha)$.

\begin{table}[t]
\caption{Estimation of the quantiles of order $\alpha/2$ and $1-\alpha/2$ ($\alpha=5\%$) for different trajectory lengths $n$, using  Algorithm 1 in Appendix B (Supplementary Materials) with $N=1\,000\,001$. 
}
\label{tab:quantiles}
\small
\centering
\begin{tabular}{ccccc}
\hline
Estimated quantiles & \multicolumn{4}{c}{ Trajectory size}\\
quantile order & 10 & 30 & 100 & asymp\\ \hline
2.5\% & 0.725 & 0.754  & 0.785 & 0.834\\ 
97.5\% & 2.626 & 2.794  & 2.873 & 2.940 \\
\hline
\end{tabular}
\end{table}

In dealing with a test, we can also be interested in computing the $p$-value. The $p$-value of the test $H_0$ \textit{vs} $H_1$ (subdiffusion as the alternative) is defined as :
\begin{equation}
p_{1,n}=F_n(T_n),
\label{eq:p-value_conf}
\end{equation}
where $F_n$ denotes the cumulative distribution function (cdf) of $T_n$ under $H_0$. The $p$-value of the test $H_0$ \textit{vs} $H_2$ (superdiffusion as the alternative) is defined as :
\begin{equation}
p_{2,n}=1-F_n(T_n).
\label{eq:p-value_dir}
\end{equation}
Testing the hypothesis $H_0$ \textit{vs} the hypotheses $H_1$ or $H_2$ is more tricky as we use a two-sided test with a non-symmetric distribution. In this case we can define the $p$-value as :
\begin{equation}
p_n=2\min\left\{ p_{1,n},p_{2,n}\right\}.
\label{eq:p-value_tot}
\end{equation}
Doubling the lowest one-tailed \Pv can be seen as a correction for carrying out two one-tailed tests.\newline
We estimate $F_n$ with the standard empirical distribution function estimated by Monte Carlo simulations using Algorithm 1 Appendix B (Supplementary Materials).
\begin{equation}
\hat{F}_n(x)=N^{-1}\sum\limits_{i=1}^N
\mathbf{1}(T^{(i)}_n\leq x).
\end{equation}
Then we estimate the $p$-value \eqref{eq:p-value_tot} substituting $\hat{F}_n$ to $F_n$.

\section{Multiple test procedure for a collection of trajectories}
\label{sec:multiple_test}
Trackers compute a collection of particle trajectories from a sequence of images.
Then, it is desirable to decide the modes of mobility for a collection of particle trajectories. 
From now, we consider a collection $\mathcal{X}_m$ of $m$ trajectories which are simultaneously observed.  
We denote by $\mathbb{X}_{n_k}^{(k)}$  the observations associated to the $k^{\text{th}}$ particle :         
\begin{align*}         
\mathbb{X}_{n_k}^{(k)}&=\left(X_{t_0}^{(k)},\ldots, X_{t_{n_k}}^{(k)}\right),\qquad k=1,\ldots,m\\         
\mathcal{X}_m &=\left\lbrace\mathbb{X}_{n_k}^{(k)},\  k=1,\ldots,m.\right\rbrace         
\end{align*}         
In this section, we denote by $\mathbb{P}$ the probability distribution of the $m$-uplet stochastic processes $\left( (X_t^{(k)}),\ k=1\ldots m\right)$ and by $\mathbb{E}$ its associated expectation.         
We assume that the observed trajectories are independent, that means $\mathbb{P}$ belongs to the tensorial product of probabilities $\mathcal{P},$ (defined in Section \ref{sec:diffusion}) $\mathbb{P}\in\mathcal{P}^{\otimes m}.$        
For all trajectories $k=1\ldots m,$ we derive our trichotomy hypothesis test procedure : $H_{0}^{(k)}$ "$(X_t^{(k)})$ is a free diffusion" versus $H_{1}^{(k)}$ "$(X_t^{(k)})$ is a subdiffusion" or $H_{2}^{(k)}$ "$(X_t^{(k)})$ is a superdiffusion".
We are faced with the problem of simultaneous tests when the rejections of null hypotheses $H_{0}^{(k)}$ are accompanied by claims of the direction of the alternative ($H_{1}^{(k)}$ or $H_{2}^{(k)}$).
In this setup, multiple test procedures are preferable than single test procedures.
Indeed, applying the procedure at level $\alpha$ for each trajectory produces in average a number of $m\alpha$ type I errors. 
A multiple testing procedure aims to control the number of false discoveries.
We refer the reader to \cite{shaffer1995multiple,roquain2011type,grandhi2015multiple} 
for a review.

A multiple testing procedure of $m$ null hypotheses against two alternative hypotheses is a rule $\mathcal{R}_1(\mathcal{X}_m)\times \mathcal{R}_2(\mathcal{X}_m),$ where  $\mathcal{R}_1(\mathcal{X}_m)$ and $\mathcal{R}_2(\mathcal{X}_m)$ are disjoint subsets of $\{H_0^{(1)},\ldots H_0^{(m)}\}.$ 
For $i=1,2,$ $\mathcal{R}_i(\mathcal{X}_m)$ is the set of the rejected  hypotheses $H_{0}^{(k)}$ to the benefit of the alternative $H_{i}^{(k)}.$
We may commit three kinds of errors in such a multiple testing procedure.
Let us introduce the following notations before listing these errors.
For a given $\mathbb{P}\in\mathcal{P}^{\otimes m},$ we denote by $\mathcal{I}(\mathbb{P})$ the subset of indexes $\{1,\ldots m\}$ for which the hypothesis $(H_{0}^{(k)})$ is actually true and by $m_0(\mathbb{P})$ the unknown cardinal of the set $\mathcal{I}(\mathbb{P}).$
We denote by $R=R_1+R_2$ the observed number of null hypotheses which are rejected by the multiple testing procedure.
Table \ref{tab:test:multiple:error} summaries the number of errors which may occur following a multiple testing procedure.
\begin{itemize}
\item We make a type I error on $H_0^{(k)}$ when we reject $H_0^{(k)}$ while it is a true null hypothesis. In this case, $k$ belongs to the set $\mathcal{I}(\mathbb{P})\cap (\mathcal{R}_1(\mathcal{X}_m)\cup\mathcal{R}_2(\mathcal{X}_m) ).$ The number of errors of first kind is $V=V_1+V_2.$
\item Type II error occurs when we do not reject a null hypothesis $H_{0,k}$ while $H_{0,k}$ is false ($k\notin\mathcal{I}(\mathbb{P})$). The number of errors of second kind is $T=T_1+T_2.$ 
\item The type III errors are directional errors : the index $k\notin\mathcal{I}(\mathbb{P})$ is correctly rejected ($k\in \mathcal{R}_1(\mathcal{X}_m)\cup\mathcal{R}_2(\mathcal{X}_m)$), but for the wrong alternative. We mix up the alternatives deciding one while it is the other. The number of errors of third kind is $S=S_3+S_4.$
\end{itemize}

\begin{table}
\centering
\small
\caption{Outcomes in testing $m$ null hypotheses against two-alternatives. For $i=1,2,$ $R_i$ is the cardinal of $\mathcal{R}_i(\mathcal{X}_m).$
The variables $(S_i)_{i=1\ldots 4}, (T_i)_{i=1,2},U,(V_i)_{i=1,2}$ are not observed and depend on $\mathcal{X}_m$ and $P.$}
\begin{tabular}{|c|ccc|c|}
\hline
\multirow{2}{*}{True situation} & \multicolumn{3}{|c|}{Decision} & \\
& Accept $H_0$ & Accept $H_1$ & Accept $H_2$ & Total\\ 
\cline{2-5}
$H_0$ & $U$ & $V_1$ & $V_2$ & $m_0(\mathbb{P})$\\ 
$H_1$ & $T_1$ & $S_1$ & $S_3$ & $m_1(\mathbb{P})$\\ 
$H_2$ & $T_2$ & $S_4$ & $S_2$ & $m_2(\mathbb{P})$\\ 
\hline
Total & $m-R_1-R_2$ & $R_1$ & $R_2$ & m\\
\hline
\end{tabular}
\label{tab:test:multiple:error}
\end{table}
To measure the type I error rate, it is common to consider the $k$-family-wise error rate (k-FWER) or the false discovery rate (FDR), see \cite{roquain2011type} and references therein. 
In our settings, controlling the type I error rate is a first step, but it would be necessary to control type III errors as well.
In the literature, the sum of the number of errors of first and third kind is controlled using the mixed-directional-family-wise error rate (mdFWER) or the mixed-directional-false discovery rate (mdFDR), see \cite{grandhi2015multiple}.
To our knowledge, the mdFWER and mdFDR are only controlled for the problem of testing null hypotheses against two-sided alternatives for finite-dimensional parameters, see for example \cite{guo2015stepwise} and references therein.

Biologists are interested in the proportions of each dynamic (subdiffusion, superdiffusion and Brownian motion) and their geographic location in the cell. 
In this context, controlling the FWER, that is the probability to make a single false discovery, is not relevant.
That is why we focus on a procedure which enables to control the FDR.
\cite[Section 5]{guo2015stepwise} also present several multiple test procedures associated to three-decision problems which aim to control the FDR.
Their approach is different since the problem is rewritten as a problem which carries out $3m$ null hypotheses.
Their proposed procedures control strongly the FDR only on $2m$ null hypotheses among the $3m$ under the dependence or independence of the test statistics.
In this section, we propose to adapt the multiple testing procedures of \cite{benjamini1995controlling} and \cite{benjamini2000adaptive} controlling the FDR that is the average proportion of false discoveries among the discoveries.
We stress that our model is non-parametric. Then we will consider the control of the mdFDR or mdFWER for a next issue.

Let $p^{(k)}, p_1^{(k)},$ and $p_2^{(k)}$ be respectively the $p$-value \eqref{eq:p-value_tot}, \eqref{eq:p-value_conf} and \eqref{eq:p-value_dir} associated to the $k^{\text{th}}$ trajectory, $k=1\ldots m.$
Let $p^{(1:m)} \leq p^{(2:m)} \leq \ldots \leq p^{(m:m)}$ be the ordered $p$-values, and $H_0^{(1:m)},\ldots H_0^{(m:m)}$ the associated null hypotheses. The adaptation of the Benjamini-Hochberg (BH) procedure is described in Procedure \ref{proc:BH}.

\begin{proc}[Adaptation of the Benjamini-Hochberg (BH) procedure]
\label{proc:BH}
\quad\newline
\vspace{-0.5cm}
\begin{enumerate}
\item Use the Benjamini-Hochberg procedure on the $p$-values $(p^{(k)})_{k=1\ldots m} :$ 
\begin{verse}
Let $k^{\star}$ be the largest $k$ for which $p^{(k:m)}\leq\frac{k}{m}\alpha$.\\
$\mathcal{R}_\alpha(\mathcal{X}_m)$ is the set of all hypotheses $H^{(k:m)}$ for $k=1,\dots,k^{\star}$.
\end{verse}
\item Let $\mathcal{R}_{1,\alpha}(\mathcal{X}_m)$ be the subset $\mathcal{R}_\alpha(\mathcal{X}_m)$ such that $p_1^{(k)}<p_2^{(k)}.$ 
\item Let $\mathcal{R}_{2,\alpha}(\mathcal{X}_m)$ be the subset $\mathcal{R}_\alpha(\mathcal{X}_m)$ such that $p_1^{(k)}>p_2^{(k)}.$ 
\end{enumerate}
\label{proc1}
\end{proc}

The set $\mathcal{R}_\alpha(\mathcal{X}_m)$ is the set of all rejected null hypotheses for our trichotomy test.
According to \cite{finner2001false}, we have,
\begin{align*}
\forall \mathbb{P}\in\mathcal{P}^{\otimes m},\quad \mathrm{FDR}(\mathcal{R}_\alpha(\mathcal{X}_m),\mathbb{P})&=\mathbb{E}\left(\frac{V}{\max(R,1)}\right)\\
&=\frac{m_0(\mathbb{P})}{m}\alpha.
\end{align*}
Then the FDR of Procedure \ref{proc:BH} is controlled by $\alpha.$
Moreover the $p$-values $p^{(k)}_{1}$ and $p^{(k)}_{2}$ give the information to which side of the distribution $F_{n_k}$ the associated test statistic $T^{(k)}_{n_k}$ is.
The case of equality ($p^{(k)}_{1}=p^{(k)}_{2}=1/2$)  never occurs since such null hypothesis will not be rejected at the step 1 of the Procedure \ref{proc1}. 

Actually, we may also use the adaptive BH procedure of \cite{benjamini2000adaptive} as the first step of Procedure \ref{proc1}. 
Then the Procedure \ref{proc1} will be referred to as the adaptive (respectively  standard) Procedure \ref{proc1} when we use the adaptive (respectively  standard) BH procedure as the first step. The adaptive BH procedure is more powerful than the standard BH procedure. 
It uses an estimation of the number of true null hypotheses $m_0(\mathbb{P})$ to increase the power of the BH procedure. 
\cite{benjamini2000adaptive} simply define the adaptive BH procedure by replacing  $m$ by an estimator $\hat{m}_0$ of $m_0$ in the BH procedure. The associated FDR is $(m_0/\hat{m}_0)\alpha$ and is less than $\alpha$ if $\hat{m}_0\leq m_0$ almost surely. 
The procedure to estimate $m_0$ presented in \cite{benjamini2000adaptive} is made for $\hat{m}_0$ to be upward biased. 
This bias favours the control of the FDR at level $\alpha$. Due to the fact that $\hat{m}_0$ does not fulfil the condition $\hat{m}_0\leq m_0$ almost surely, we can not say that the adaptive BH procedure controls the FDR at level $\alpha$ theoretically. 
However simulations from \cite{benjamini2000adaptive} suggest that the adaptive BH procedure controls the FDR at level $\alpha$.


\section{Simulation study and real data applications}
\label{sec:exp-res}
We assess the power of our single test procedure (on a single trajectory) and our multiple test procedure (on a collection of trajectories) by Monte Carlo simulations. 
We consider parametric alternatives : the \OU \eqref{eq:sde_OU} and the fractional Brownian motion with Hurst index $0<\mathfrak{h}<1/2$  for subdiffusion processes ($H_1$); the \BRD \eqref{eq:sde_dir} and the fractional Brownian motion with Hurst index  $1/2<\mathfrak{h}<1$ for superdiffusion processes ($H_2$). 
Then, we apply our procedure on real data comparing our results with those obtained thanks to a method based on the mean square displacement.

\subsection{Power of the test procedure for a single trajectory}

In Section \ref{sec:test}, we study the asymptotic distribution of the test statistic under the null hypothesis and parametric alternative hypotheses. 
More precisely Proposition \ref{prop1} states that the power of the test under parametric alternatives converges to $1$ with $n.$ 
Figure \ref{fig:power_single} shows the Monte Carlo estimates of the power under the parametric alternatives aforementioned in Proposition \ref{prop1}.
For a fixed step of time $\Delta$ and a fixed diffusion coefficient $\sigma,$ we vary the values of the other parameters and the length $n$ of the trajectories. For each parametric alternatives of Proposition \ref{prop1}, we can use exact simulation schemes.\\

\newpage

\begin{minipage}{\linewidth}
\begin{tabular}{cc}
\qquad \qquad  (a) & \qquad \qquad (b) \\
 \hspace{-1.5cm} \pgfplotstableread {fig/power_dbr.dat} {\loadedtable}
\pgfplotsset{
    compat=newest,
    every axis/.append style={
        legend image post style={xscale=0.5}
    }
}
\begin{tikzpicture}[baseline]
\begin{axis}
[xlabel=$v$,ylabel=Power,
minor tick num=4,
enlarge x limits=false,
enlarge y limits=false,
width=0.5\linewidth,
legend style={
   at={(0.6,0.05)}, anchor=south west, legend columns=1,font=\tiny}
] 

\addplot[color=red, dashed] table[x=dim_less_dir,y=p_dbr_10] from \loadedtable;
\addplot[color=blue, dash dot] table[x=dim_less_dir,y=p_dbr_30] from \loadedtable;
\addplot[color=green] table[x=dim_less_dir,y=p_dbr_50] from \loadedtable;

\legend{$n=10$,$n=30$,$n=50$}

\end{axis}
\end{tikzpicture}  & \pgfplotstableread {fig/power_ou.dat} {\loadedtable}
\pgfplotsset{
    compat=newest,
    every axis/.append style={
        legend image post style={xscale=0.5}
    }
}
\begin{tikzpicture}[baseline]
\begin{axis}
[xlabel=$\lambda$,ylabel=Power,
minor tick num=4,
enlarge x limits=false,
enlarge y limits=false,
width=0.5\linewidth,
legend style={
   at={(0.6,0.45)}, anchor=south west, legend columns=1,font=\tiny}
] 

\addplot[color=red, dashed] table[x=dim_less_ou,y=p_ou_10] from \loadedtable;
\addplot[color=blue, dash dot] table[x=dim_less_ou,y=p_ou_30] from \loadedtable;
\addplot[color=green] table[x=dim_less_ou,y=p_ou_50] from \loadedtable;

\legend{$n=10$,$n=30$,$n=50$}

\end{axis}
\end{tikzpicture}  \\
  \multicolumn{2}{c}{   (c)}\\
 \multicolumn{2}{c}{\hspace{-1.5cm}\pgfplotstableread {fig/power_fbm.dat} {\loadedtable}
\pgfplotsset{
    compat=newest,
    every axis/.append style={
        legend image post style={xscale=0.5}
    }
}
\begin{tikzpicture}
\begin{axis}
[xlabel=$\mathfrak{h}$,ylabel=Power,
minor tick num=4,
enlarge x limits=false,
enlarge y limits=false,
width=0.5\linewidth,
legend style={
   at={(0.4,0.6)}, anchor=south west, legend columns=1,font=\tiny}
] 

\addplot[color=red, dashed] table[x=H,y=p_fbm_10] from \loadedtable;
\addplot[color=blue, dash dot] table[x=H,y=p_fbm_30] from \loadedtable;
\addplot[color=green] table[x=H,y=p_fbm_50] from \loadedtable;

\legend{$n=10$,$n=30$,$n=50$}

\end{axis}
\end{tikzpicture}} 
\end{tabular}
\vspace{-0.5cm}

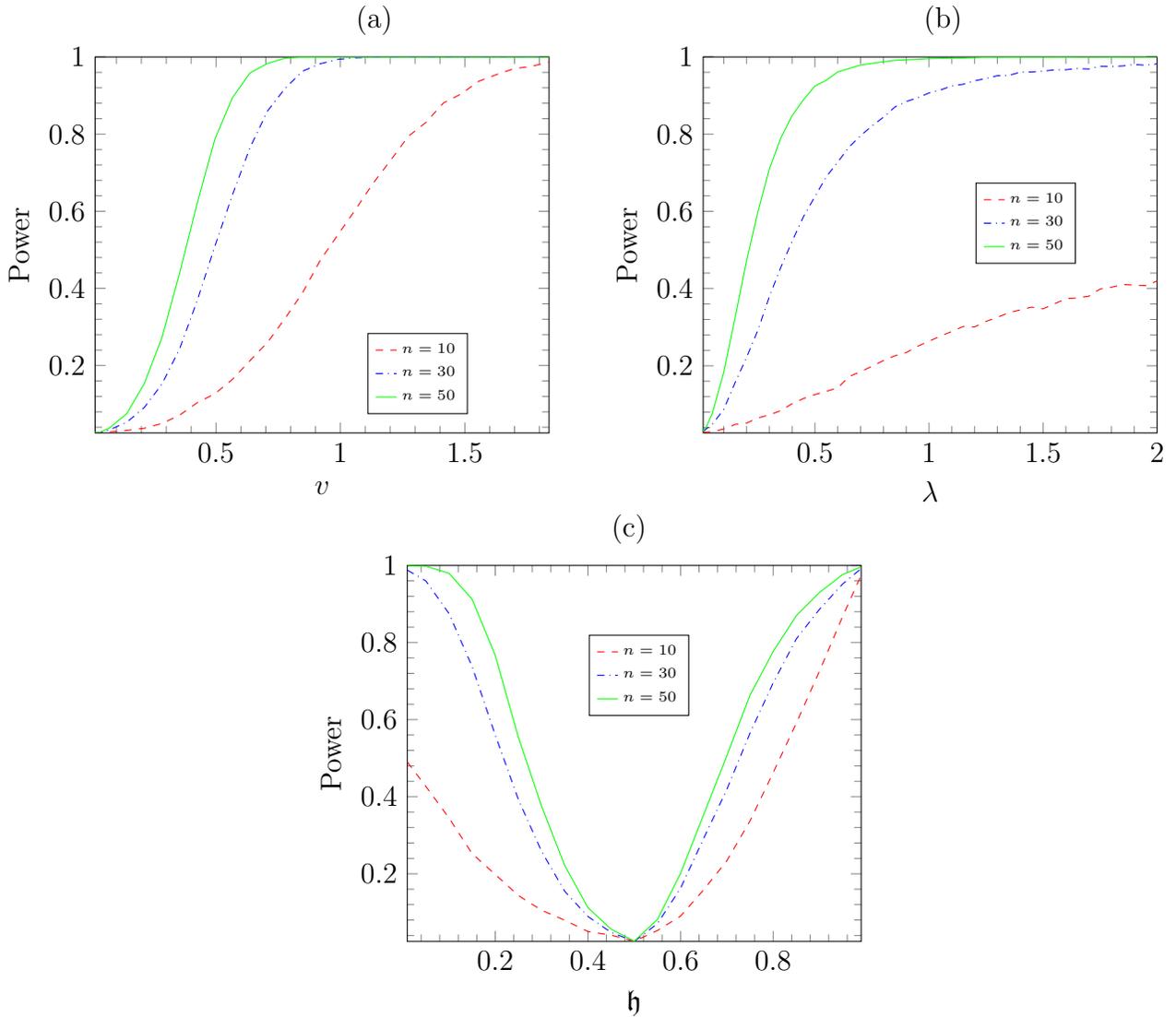
\captionof{figure}{Monte Carlo estimate of the power of the test at level $\alpha=0.05$ according to the trajectory length $n$ and the parameter associated to the following parametric alternatives : (a) \BRD (parameter $v=(v_1,v_2)$ such that $v_1=v_2$); (b) the Ornstein-Uhlenbeck process (parameter $\lambda$) and (c) fractional Brownian motion (parameter $\mathfrak{h}$).  We use $10\,001$ Monte Carlo replications for computing each point of the power curves.}
\label{fig:power_single}
\end{minipage}

If $(X_t)$ is an Ornstein-Uhlenbeck process \eqref{eq:sde_OU} which is entered in its stationary regime, then the distribution of the test statistic does not depend on $\theta$ (see Appendix A.4 (Supplementary Material)).
Figure \ref{fig:power_single}(b) shows the plot of the power regarding the values of $\lambda$ which models the strength of the restoring force toward the equilibrium position $\theta$. 
Stronger is the force, more powerful is the test.
 
Furthermore if $(X_t)$ is a \BRD with parameters $(v,\sigma)$ such that $\norme{v}\sqrt{\Delta}>\sigma$, then the particle goes toward the direction of $v$ while the Brownian random part of the SDE \eqref{eq:sde_dir} does not affect much its trajectory (see Appendix A.4 (Supplementary Material)). The bigger is the norm of the drift parameter $v$, more powerful is the test, see Figure \ref{fig:power_single}(a). 

Finally if $(X_t)$ is a fractional Brownian motion, then the distribution of $T_n$ depends only on the Hurst index $\mathfrak{h}$ (see Appendix A.4 (Supplementary Material)). Then the test procedure is equivalent to test the null hypothesis "$\mathfrak{h}=1/2$" versus "$\mathfrak{h}\neq 1/2$", see Figure \ref{fig:power_single}(c).

\subsection{The Average Power and the mdFDR of the multiple test procedure for a collection of trajectories}

The simulation settings are described as follows. 
According to experience, we choose the number of trajectories to be $m=100$ or $m=200.$
All trajectories are assumed to have the same size $n=30,$ since this size is reasonable regarding real data.
The diffusion coefficient $\sigma$ and the lag-time $\Delta$ are set to $1.$
The collection of trajectories $\mathcal{X}_m$ is composed of :
\begin{itemize}
\item $m_0<m$ Brownian trajectories ($H_0$);
\item $(m-m_0)/2$ subdiffusive trajectories ($H_1$), half from an Ornstein-Uhlenbeck process with parameter $\lambda>0$, half from a fractional Brownian motion with Hurst index $0<\mathfrak{h}<1/2$;
\item $(m-m_0)/2$ superdiffusive trajectories ($H_2$), half from a Brownian motion with drift $v \in \mathbb{R}^2$, half from a fractional Brownian motion with Hurst index $1/2<\mathfrak{h}<1$.
\end{itemize}
The parameters to simulate these trajectories are given in Table \ref{tab:param_val}.
We take the parameters corresponding to a power of the single test procedure of 80\%. Such parameters are used to produce Figure \ref{fig:typical_traj} (a). This choice seems coherent in regards to trajectories from real data, see Figure \ref{fig:typical_traj} (b).
For a given $m$, the proportion of true null hypotheses $H_0$ varies : $m_0/m\in\{0,0.2,0.4,0.6,0.8\}$.

\begin{table}[t]
\centering
\caption{ Parameters used for simulating the alternative hypotheses. For simplicity we took $\sigma=1$ for all processes (including Brownian motion). We choose $\Delta=1$.
}
\begin{tabular}{llcc}
\hline
Hypothesis&Process & Parameter &Value \\
\hline 
$H_1$ &\OU & $\lambda$ & 0.53 \\ 
$H_1$ &Fractional  Brownian & $\mathfrak{h}$ & 0.13 \\ 
$H_2$ &\BRD & $\norme{v}$ & 0.66 \\ 
$H_2$ &Fractional  Brownian & $\mathfrak{h}$ & 0.85 \\ 
\hline
\end{tabular} 
\label{tab:param_val}
\end{table}

The mdFDR is a rate which controls the error of type I and type III. 
It is defined as $\mathbb{E}((V+S)/\max(R,1))$ (see Table \ref{tab:test:multiple:error}). 
Table  8 Appendix B (Supplementary Materials)  shows that the Procedure \ref{proc1} also controls the mdFDR. 
The mdFDR and FDR appear to be very close meaning that the number of type III errors is extremely low.
Furthermore, the adaptive Procedure \ref{proc1} (where $m_0$ is estimated) is less conservative than the standard Procedure \ref{proc1}.
As expected, the FDR and mdFDR increase as the proportion of true null hypotheses increases.

To assess the performance of our multiple test procedure, we use the average power \citep{grandhi2015multiple} :
\begin{equation}
\mathbb{E}\left(\frac{S_i}{m_i}\right), \quad i=1,2
\end{equation}
where $m_i$ is the number of true alternatives $H_i$ and $S_i$ $(i=1,2)$ is defined in Table \ref{tab:test:multiple:error}. In our simulation scheme, we set $m_i=(m-m_0)/2$. The average power is the expected proportion of hypotheses accepted as $H_i$ among all true alternatives $H_i$. Average powers of the different simulations corresponding to different values of $m_0/m$ and $m$ are shown on Figure \ref{fig:multiple_power}.\newline
First, we can see that the powers of $H_1$ and $H_2$ are not very sensitive to the number of hypotheses $m$ for both the standard Procedure \ref{proc1} and the adaptive Procedure \ref{proc1}. Secondly, the adaptive Procedure \ref{proc1} is more powerful than the standard Procedure \ref{proc1} (red and blue dashed lines respectively above red and blue solid lines in Figure \ref{fig:multiple_power}). The benefit of the adaptive Procedure \ref{proc1} over the standard Procedure \ref{proc1} decreases as the proportion of true null hypotheses $m_0/m$ increases (solid and dashed line of same color getting closer as $m_0/m$ increases in Figure \ref{fig:multiple_power}). This is due to the fact that, as $m_0/m$ tends to 1, $m_0$ and then $\hat{m}_0$ tend to $m$.  As a result, the adaptive and standard Procedure \ref{proc1} become similar.

\begin{rem}
We observe that, given a certain procedure (standard or adaptive Procedure \ref{proc1}), the average power of $H_1$ is lower than the average power of $H_2$, see Figure \ref{fig:multiple_power}. It is not due to the choice of parameters as both alternatives $H_1$ and $H_2$ are simulated to share the same power (80\%) with the single test procedure. Actually, it comes from the fact that the \Pvs under $H_2$ are stochastically smaller than the \Pvs under $H_1$ (see Appendix C Figure 6 (Supplementary Material)). Then, the true superdiffusive trajectories are more easily detected as non Brownian in the first step of the (adaptive) Procedure \ref{proc1} than the true subdiffusive trajectories. We note that, if we use other parametric models for subdiffusion ($H_1$) and superdiffusion ($H_2$), we can have the opposite situation.
\end{rem}

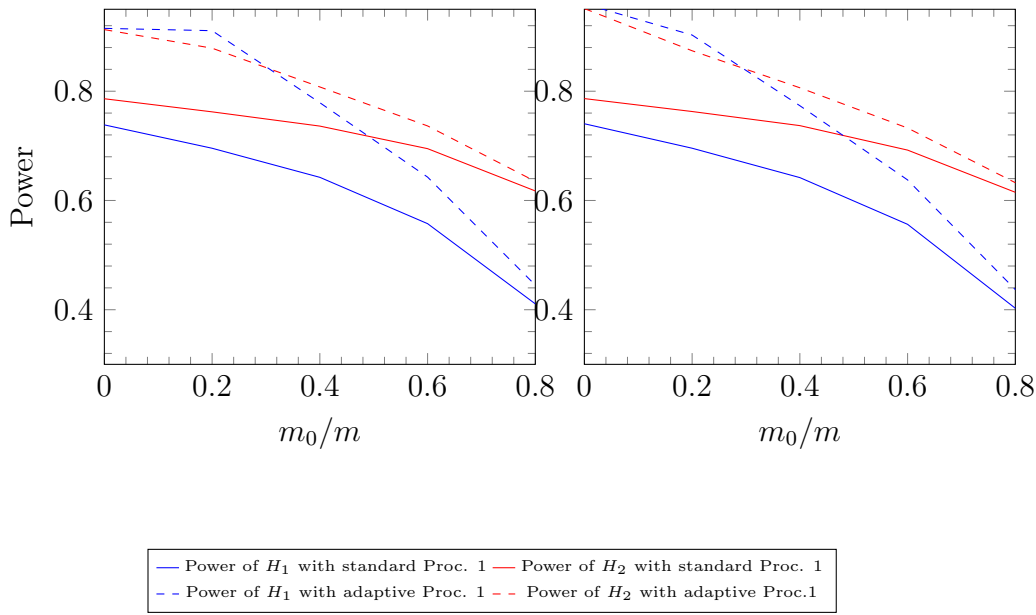
\begin{figure}[h!]
\centering
\begin{tabular}{cc}
   \pgfplotstableread {fig/power_benj_100.dat} {\loadedtable}
\pgfplotsset{
    compat=newest,
    every axis/.append style={
        legend image post style={xscale=0.5}
    }
}
\begin{tikzpicture}[baseline]
\begin{axis}
[xlabel=$m_0/m$,ylabel=Power,
minor tick num=4,
enlarge x limits=false,
enlarge y limits=false,
ymin=0.3,
ymax=0.95,
width=0.45\linewidth,
legend style={
cells={align=left},at={(0.1,-0.7)}, anchor=south west, legend columns=2,font=\tiny}
] 

\addplot[color=blue] table[x=prop_h0,y=p_sub_benj] from \loadedtable;
\addplot[color=red] table[x=prop_h0,y=p_sup_benj] from \loadedtable;
\addplot[color=blue,dashed] table[x=prop_h0,y=p_sub_benj_ad] from \loadedtable;
\addplot[color=red,dashed] table[x=prop_h0,y=p_sup_benj_ad] from \loadedtable;
\legend{Power of $H_1$ with standard Proc. 1, Power of $H_2$ with standard Proc. 1, Power of $H_1$ with adaptive Proc. 1,Power of $H_2$ with adaptive Proc.1, Power of a single test procedure \\ on a single trajectory from $H_1$ or $H_2$}

\end{axis}
\end{tikzpicture} & 
   \hspace{-5cm}\pgfplotstableread {fig/power_benj_200.dat} {\loadedtable}
\pgfplotsset{
    compat=newest,
    every axis/.append style={
        legend image post style={xscale=0.5}
    }
}
\begin{tikzpicture}[baseline]
\begin{axis}
[xlabel=$m_0/m$,
minor tick num=4,
enlarge x limits=false,
enlarge y limits=false,
ymin=0.3,
ymax=0.95,
width=0.45\linewidth,
legend style={
cells={align=left},at={(-0.9,-0.1)}, anchor=south west, legend columns=2,font=\tiny}
] 

\addplot[color=blue] table[x=prop_h0,y=p_sub_benj] from \loadedtable;
\addplot[color=red] table[x=prop_h0,y=p_sup_benj] from \loadedtable;
\addplot[color=blue,dashed] table[x=prop_h0,y=p_sub_benj_ad] from \loadedtable;
\addplot[color=red,dashed] table[x=prop_h0,y=p_sup_benj_ad] from \loadedtable;


\end{axis}
\end{tikzpicture}\\
\end{tabular}
 \caption{Monte Carlo estimate of the average power against the proportion of true null hypothesis $m_0/m$ in the collection of hypotheses. On the left we test $m=100$ hypotheses, on the right $m=200$. 
 }
 \label{fig:multiple_power}
\end{figure}

Finally, we compare the adaptive Procedure \ref{proc1} to the MSD classification of \cite{feder1996constrained}, based on a fit of the MSD curve to $t\rightarrow t^\beta$ , see Section \ref{subsec:msd}. We assess the two methods on a single collection of trajectories $\mathcal{X}_m$ with $m=200$ and $m_0/m=0.4$, composed of a mixture of Brownian motion, subdiffusion and superdiffusion as described at the beginning of this section. We get the confusion matrices Table \ref{tab:confusion:ad_proc1} and \ref{tab:confusion:msd} for respectively the adaptive Procedure \ref{proc1} and the MSD method. As suggested by the limiting curves used by \cite{feder1996constrained} (see Figure \ref{fig:msd_curve}), the MSD method mixes up the Brownian trajectories with both subdiffusion and superdiffusion (see line 1 of Table \ref{tab:confusion:msd}). Another big issue is that 40\% of the particles undergoing subdiffusion are considered as immobile by the MSD method. On the other hand, the adaptive procedure \ref{proc1} detects well subdiffusion and superdiffusion in the setting of this simulation (line 2 and 3 of Table \ref{tab:confusion:ad_proc1}). More importantly, it controls the number of false discoveries through the FDR (line 1 of  Table \ref{tab:confusion:ad_proc1}).

\begin{table}[t!]
\centering
\caption{Confusion matrix for the MSD method}\label{tab:confusion:msd}
\begin{tabular}{lcccc}
\hline 
Ground truth/Test label & Brownian & Subdiffusion & Superdiffusion& Not moving \\ 
\hline 
Brownian & 19  & 45 & 36 & 0 \\ 
Subdiffusion & 0 & 60 & 0 & 40 \\ 
Superdiffusion & 3 & 0 & 97 &0 \\ 
Not moving & 0 & 0 & 0 & 0 \\ 
\hline 
\end{tabular} 
\end{table}

\begin{table}[t!]
\centering
\caption{Confusion matrix for the adaptive Proc.1}\label{tab:confusion:ad_proc1}
\begin{tabular}{lccc}
\hline 
Ground truth/Test label & Brownian & Subdiffusion & Superdiffusion \\ 
\hline 
Brownian & 96  & 0 & 4 \\ 
Subdiffusion & 23 & 77 & 0 \\ 
Superdiffusion & 10 & 0 & 90 \\ 
\hline 
\end{tabular} 
\end{table}

\subsection{Real data : the Rab11a protein sequence}
Fluorescence imaging and microscopy has a prominent role in life science and medical research. It consists of detecting specific cellular and intracellular objects of interest at the diffraction limit (200 nm).  
%
%
These objects are first tagged with genetically engineered proteins that emit fluorescence. Then, they can be observed   using wide field or confocal microscopy. Several image analysis methods have been  developed to quantify intracellular trafficking, including object detection and tracking of fluorescent tags in cells (\cite{chenouard2014objective,Kervrann2016}).

Here, we are particularly interested in studying the exocytosis process, that is the mechanism of active transport of proteins out of the cell.  Small structures, called the vesicles, travel from organelles to the cell membrane, propelled by motor activity. The vesicle fuses with the plasma membrane and delivers the transported protein in the extra-cellular medium. Given computed trajectories, we investigate here the quantification of vesicles dynamics  and trafficking. As explained earlier in the paper, the trajectories can be generally classified into three categories : Brownian motion, subdiffusion and superdiffusion.

As a model of exocytosis/recycling, we focus on the Rab11a protein. This protein is a member of the dynamic architecture of the complex molecular assembly which regulates recycling organelles trafficking. It plays an essential role in the regulation of late steps of vesicle recycling to the plasma membrane, namely the tethering-docking process (\cite{Schafer2014}). During exocytosis, Rab11a is attached to  the vesicle membrane. Then, tracking Rab11a amounts to tracking the vesicle during the exocytosis phase. After the fusion of the vesicle to the cell membrane, Rab11a is recycled in the cytosol. During the recycling step, the tracking of Rab11a is not accurate as the proteins are detached from the vesicle and scatter around the cytosol. It is currently under investigation. For that reason, we focus on the exocytosis process until the fusion time with the cell membrane. 

An illustration of the Rab11a sequence is shown in Figure \ref{fig:rab11_data}  where the dark spots correspond to Rab11a vesicles in a ``crossbow'' micro-patterned shape cell.
{A typical image extracted from an image sequence is shown Figure \ref{fig:rab11_data}. The image sequence is composed of 600 images of size $256 \times 240$ (1 pixel=160nm) acquired at 10 frames/s ($\Delta=0.1s)$. 
We tracked $1\,561$ trajectories  with the multiple hypothesis tracking method with default parameters (\cite{chenouard2013multiple}), available on the Icy software (\url{http:www.icy.org}). However, we discarded too small and to long trajectories corresponding to tracking errors in most cases.
Then, we have to get rid of the particles that do not move enough and consequently, can not be modelled by diffusion processes. In practice, we analyse only the trajectories with at least 20 distinct positions and  the vesicles that stop at the same position less than  $K = \lfloor n/10 \rfloor$ times  (with $n$ the length of the trajectory). In the case of the aforementioned image sequence, we end up with 166 trajectories whose median length is $n = 83$. 

In Figure \ref{fig:rab11_data}, our results show that the four procedures -- adaptive Procedure \ref{proc1}, standard Procedure \ref{proc1}, single test and  MSD method -- do not produce similar classification results visually. 
From the simulations, we found that the MSD method tends to wrongly over-detect subdiffusion and superdiffusion (see Tables \ref{tab:confusion:msd} and \ref{tab:confusion:ad_proc1}). This is probably true also in the case of real Rab11a sequence. 
In Table \ref{tab:rab11_data}, we give the proportion of each type of diffusion for the different methods aforementioned. 
The adaptive Procedure \ref{proc1} tends to decrease the number of Brownian trajectories compared to the standard Procedure \ref{proc1}. It is not surprising as the adaptive Procedure \ref{proc1} is defined to be more powerful than the standard Procedure \ref{proc1} : it rejects more easily the null hypothesis. This gain in power benefits to the alternative $H_1$ (subdiffusion). In fact we detect 23\% of subdiffusion for the adaptive Procedure \ref{proc1} against 16\% for the standard Procedure \ref{proc1} while both detect 4\% of superdiffusion (see Table \ref{tab:rab11_data}). The single test procedure detects even less Brownian motion but we know that it can not control the FDR.
In Figure \ref{fig:rab11_data}, the subdiffusion trajectories labelled with the test approach are more located in the center of the cell  in a region corresponding to the Endosomal Recycling Compartment which is known to organize Rab11a carrier vesicles (\cite{Schafer2014}). It is also true for the subdiffusion trajectories labelled with the MSD analysis but we have just said that there is probably an over-detection of the subdiffusion with this method. We note that we carry the classification of trajectories with our different test procedures and the MSD method on multiple sequences of Rab11a protein, see Appendix C Figure 7 (Supplementary Material).
}

\newpage
\begin{minipage}{0.95\linewidth}
\centering
\begin{tabular}{cc}
(a) & (b)\\
\includegraphics[width=0.5\textwidth]{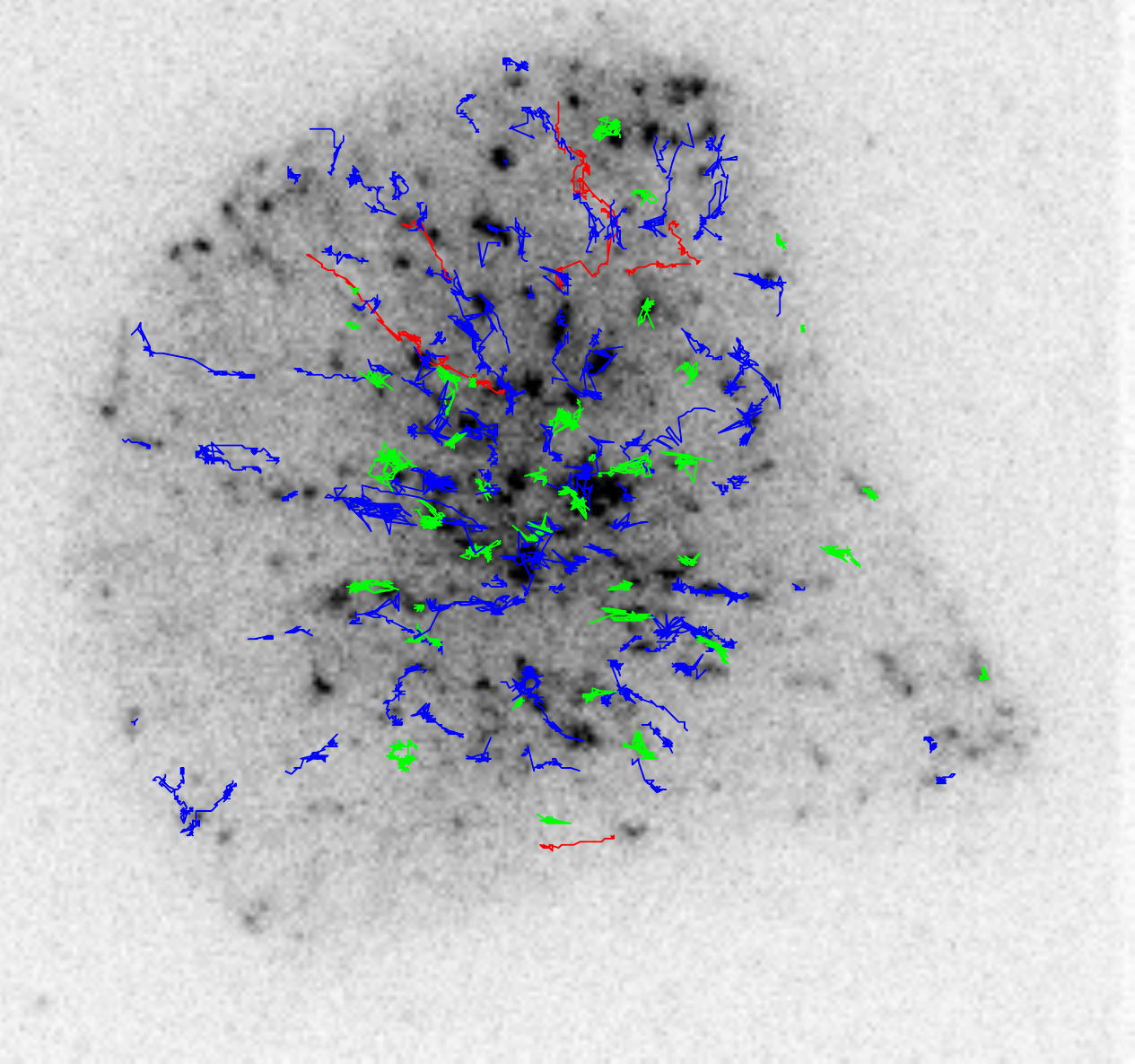} 
&
\includegraphics[width=0.5\textwidth]{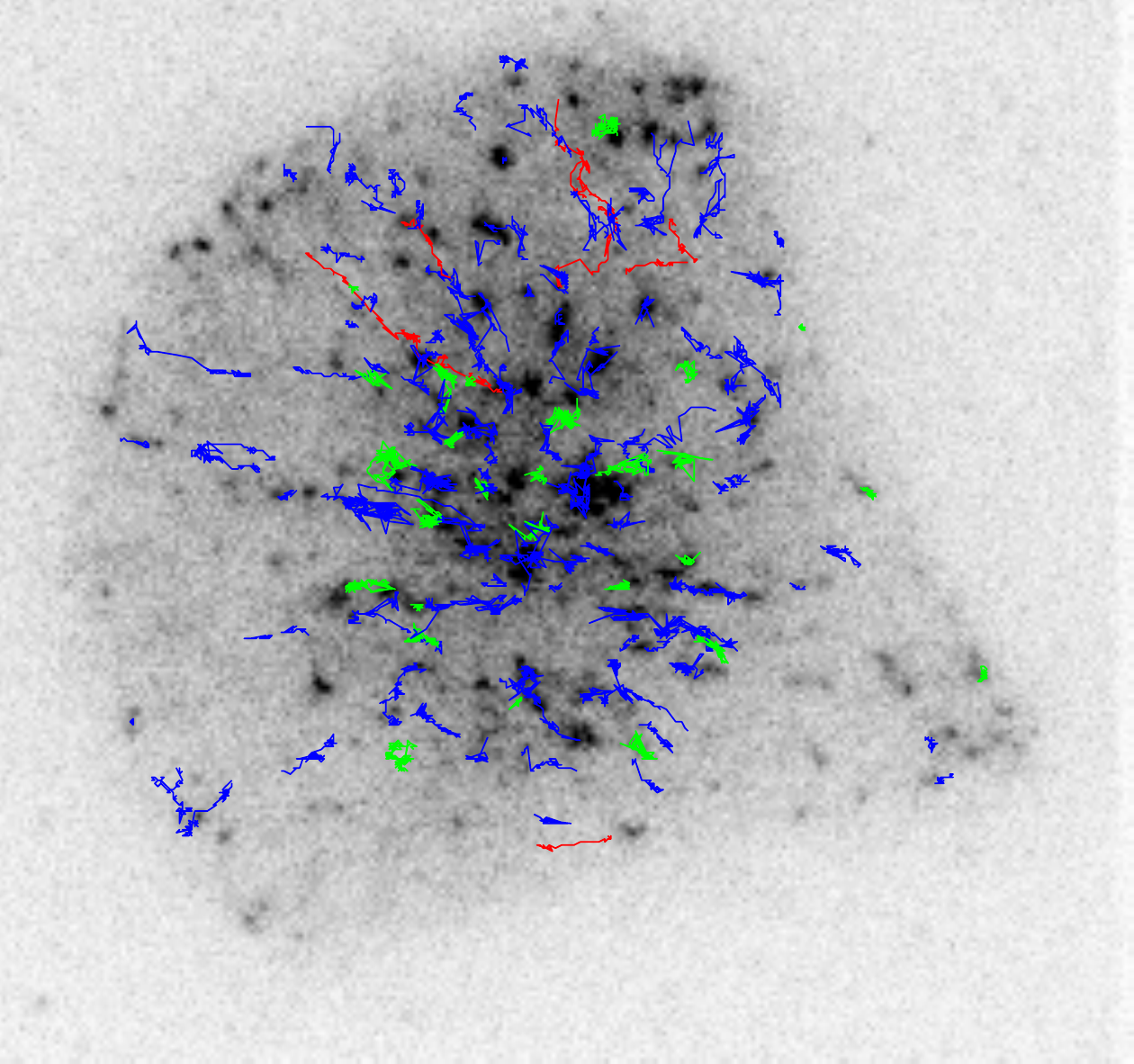}  
\\
(c) & (d)\\
\includegraphics[width=0.5\textwidth]{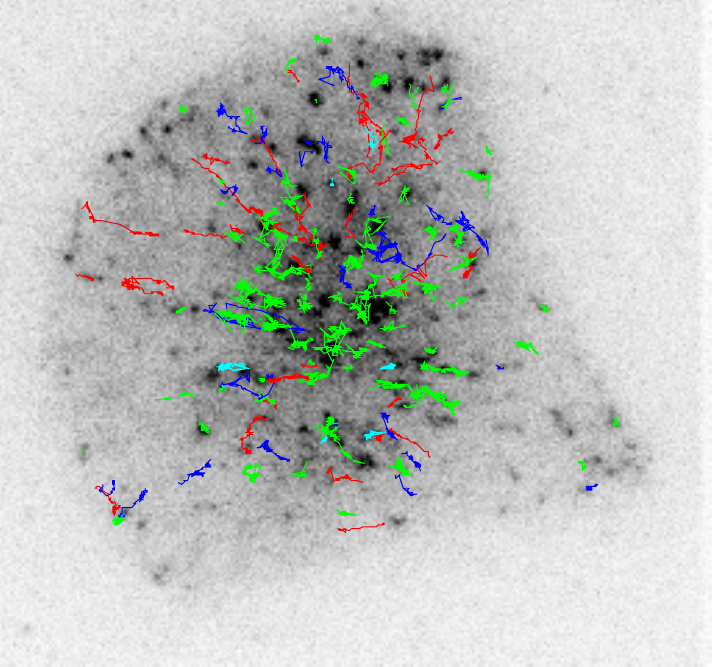}
&
\includegraphics[width=0.5\textwidth]{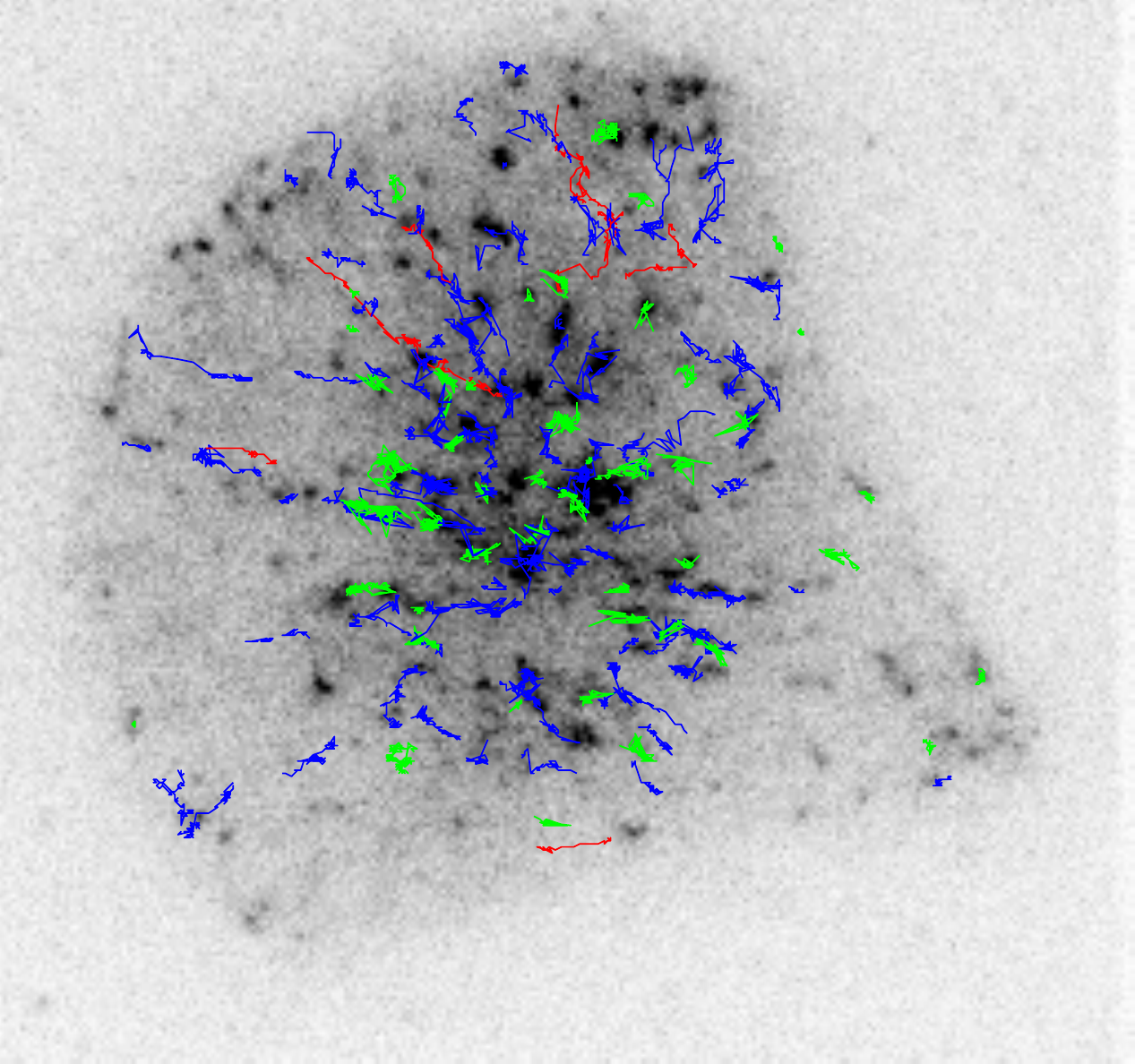}  
\end{tabular}
\vspace{-0.5cm}

 \captionof{figure}{Map of the classification of the trajectories of the Rab11a sequence with (a) standard multiple test procedure \ref{proc1}, (b) its adaptative version, (c) MSD, (d) single test procedure. The colour code is: blue for Brownian motion, red for superdiffusion and green for subdiffusion, cyan for immobile particule (for the MSD method only).}
 \label{fig:rab11_data}
\end{minipage}%

\begin{table}[t!]
\centering
\caption{Percentages of Brownian, superdiffusive and subdiffusive trajectories in the Rab11a sequence according to the different methods of classification.}\label{tab:rab11_data}
\begin{tabular}{lccc}
\hline 
Method & Brownian & Subdiffusion & Superdiffusion \\ 
\hline 
Standard Proc. 1 & 80  & 16 & 4 \\ 
adaptive Proc. 1 & 73 & 23 & 4 \\ 
Single test & 66 & 28 & 6 \\ 
MSD & 16 & 63 & 21 \\ 
\hline 
\end{tabular} 
\end{table}
\section{Discussion}
In this paper, we proposed a method for classifying the particle trajectories observed in living cells into three types of diffusion: Brownian motion, subdiffusion and superdiffusion. We used a test approach with the Brownian motion as the null hypothesis. More specifically, we developed a non-parametric three-decision test whose alternatives are subdiffusion and superdiffusion. On the one hand we built a single test procedure for testing a single trajectory, on the other hand we proposed a multiple test procedure for testing a collection of trajectories. These procedures control respectively the type I error and the false discovery rate at level $\alpha$. It is worth noting that the length of the trajectory $n$ is taken into account in our classification rule. Our approach can be considered as an alternative to the MSD method. It gives more reliable results as confirmed by our Monte Carlo simulations and evaluations on real sequences of images depicting protein dynamics acquired with TIRF or SPT-PALM microscopy.

\section*{Source code}

A Matlab package of the method is available at: \newline
\textcolor{blue}{{\url http://serpico.rennes.inria.fr/doku.php?id=software:thot:index}}

\section*{Acknowledgements}
We thank Jean Salamero (UMR 144 CNRS-Institut Curie) who provided the microscopy image sequences and for his helpful insights and assistance with experiments.

\bibliographystyle{agsm}

\bibliography{biblio}


\appendix
\newpage

\section{Proofs}\label{sec:proofs}
\subsection{Proof of Theorem \ref{thm:conv:H0}}
\begin{proof}[Proof of Theorem \ref{thm:conv:H0}]\label{app:proof_thm1}
Under the null hypothesis, $X_t/\sigma= B_t$ is a standard Brownian Motion.
Let us introduce the following random variable,
\begin{equation}\label{app:eq:wiener}
\tilde{T}_n = \max_{k=1\ldots n}\left\|\frac{1}{\sqrt{n}} R_k\right\|_2, 
\end{equation}
where $R_k=\sum_{j=1}^k (B_{j\Delta} - B_{(j-1)\Delta})/{\sqrt{\Delta}}.$
Since $\hat{\sigma}_n$ is a consistent estimator of $\sigma$ and using the Slutsky Lemma, it remains to prove that $\tilde{T}_n$ converges in distribution to $S_0.$
Using the fact that the increments of the Brownian process are independent and Gaussian, $R_k$ is the sum of $j$ independent identically $\mathcal{N}(0,1)$-distributed random variables. 
We define the following process,
$$
W_t^{(n)} = \frac{1}{\sqrt{n}} R_{\lfloor nt\rfloor},\qquad t\in[0,1],
$$
where $\lfloor x\rfloor$ denotes the integer part of $x\in\mathbb{R}.$ 
Then we get:
\begin{equation}
\tilde{T}_n=\sup_{t\in[0,1]}\left\| W_t^{(n)} \right\|_{2}.
\end{equation}

Due to Donsker's Theorem \citep[Theorem 8.2]{billingsley2013convergence}, $(W_t^{(n)})$ converges in distribution to the Wiener measure as $n\rightarrow \infty$ over the space of continuous function on $[0,1].$ 
Since $x\to\sup_{t\in[0,1]}\|x(t)\|$ is a continuous function on the space of continuous functions from $[0,1]$ to $\mathbb{R},$ $\tilde{T}_n$ converges in distribution to $S_0$. 
\end{proof}

\subsection{Proof of Proposition \ref{prop:sigma1} : the convergence of the estimator \eqref{eq:sigma_standard} of the diffusion coefficient}
\label{app:proof_sigma}
Notice that $\hat{\sigma}_n=\hat{\sigma}_{1,n}$ is strongly consistent under the null hypothesis due to the strong law of large numbers and the independence of the increments of the Brownian motion.

We focus now on the three alternatives.
According to the alternative, we denote by $\mathbb{E}$ the expectation associated to the measure $P$ of the solution of the related SDE (\eqref{eq:sde_FBM} or \eqref{eq:sde_OU} or \eqref{eq:sde_dir}).

\begin{proof}[Brownian with drift]
We may rewrite the strong solution of the SDE \eqref{eq:sde_dir} as,
\begin{equation*}
X_{t_{k}}=X_{t_{k-1}}+v \Delta+\sigma \sqrt{\Delta} \epsilon_{k},\qquad k=1\ldots n,
\end{equation*}
where $\sqrt{\Delta}\epsilon_k=B_{t_{k}}-B_{t_{k-1}},$ and $(B_t)$ is a standard Brownian motion.
Then the random variables $Z_k=\|v \Delta+\sigma \sqrt{\Delta} \epsilon_{k}\|^2,$ $k=1\ldots n,$ are positive independent identically distributed random variables, and admit a moment of order 1,
$$ 
\mathbb{E}(Z_k)=\Delta^2\|v\|^2+2\Delta\sigma^2.
$$
Then according to the strong law of large numbers, $\hat{\sigma}_n$ converges almost surely to ${\Delta\|v\|^2}/{2}+\sigma^2.$
\end{proof}

\begin{proof}[Ornstein-Uhlenbeck process]
Let $(X_t)$ be an Ornstein-Uhlenbeck process  \eqref{eq:sde_OU}. 
The SDE \eqref{eq:sde_OU} admits a unique solution \citep[Section 2.2.3]{bressloff2014stochastic}
\begin{equation}
X_t-X_{s}=(X_{s}-\theta)(e^{-\lambda (t-s)}-1)+\sigma \int_{s}^t e^{-\lambda (t- u)} dB^{1/2}_u.
\label{app:eq:ou_sde_int}
\end{equation}
Then $(X_t)$ is a stationary Gaussian process where transition density $p(s,x,t,y)$ is the density of 
$$
\mathcal{N}\left(x+(x-\theta)(e^{-\lambda (t-s)}-1),\sigma^2(1-e^{-2\lambda (t-s)})/(2\lambda)\mathbf{I}_2\right).
$$
Then we get that,
\begin{align*}
\mathbb{E}(\|X_{t+\Delta}-X_t\|^2\mid X_t=x) &= \int\|x-y\|^2 p(t,x,t+\Delta,y)dy,\\
&= \|x-\theta\|^2(e^{-\lambda\Delta}-1)^2 + \sigma^2(1-e^{-2\lambda\Delta})/\lambda.
\end{align*}
Moreover the density $\mu$ of the stationary distribution of $(X_t)$ is the Gaussian variable $\mathcal{N}\left(\theta,(\sigma^2\mathbf{I}_d)/(2\lambda)\right).$
Then we obtain that,
\begin{align*}
\mathbb{E}(\|X_{t+\Delta}-X_t\|^2) &= \int \mathbb{E}(\|X_{t+\Delta}-X_t\|^2\mid X_t=x) \mu(x)dx,\\
&= \sigma^2(e^{-\lambda\Delta}-1)^2/\lambda + \sigma^2(1-e^{-2\lambda\Delta})/\lambda,\\
&= 2\sigma^2(1-e^{-\lambda\Delta})/\lambda.
\end{align*}
Now, according to \citet[Lemma 3.1]{bibby1995martingale}, if $(X_t)$ is a stationary diffusion, $\hat{\sigma}_n^2$ converges in probability to $\mathbb{E}(\|X_{t+\Delta}-X_t\|^2)/(2\Delta).$ We deduce the result.
\end{proof}

\begin{proof}[Fractional Brownian Motion]
Let $(X_t)$ be a fractional Brownian motion \eqref{eq:sde_FBM}. 
Due to the self-similarity property and the stationary increments of the fractional Brownian motion, the following process,
$$
W_t^{(n)}= \frac{X_{t_0+n\Delta t}-X_{t_0}}{(n\Delta)^{\mathfrak{h}}\sigma},\qquad t\in[0,1],
$$
is a standard fractional Brownian motion.
The statistic associated to the quadratic variation of the process $(W_t^{(n)})$ may be defined as,
\begin{align*}
V_n &= \frac{1}{n}\sum_{i=1}^n \frac{\left\|W_{i/n}^{(n)}-W_{(i-1)/n}^{(n)}\right\|^2}{\mathbb{E}\left\|W_{i/n}^{(n)}-W_{(i-1)/n}^{(n)}\right\|^2} - 1,\\
&= \frac{\hat{\sigma}_n^2}{\sigma^2\Delta^{2\mathfrak{h}-1}} -1.
\end{align*}
According to \citet[Proposition 1]{coeurjolly2001estimating}, $V_n$ converges almost surely to $0.$ Then we deduce that $\hat{\sigma}_n^2/\sigma^2$ tends to $\Delta^{2\mathfrak{h}-1}$ almost surely.
\end{proof}

\subsection{Proof of Proposition \ref{prop1} : the asymptotic behaviour of the test statistic under parametric alternatives}
\label{app:proof_conv_h1_h2}
Since the diffusion parameter $\sigma$ is unknown, the test statistic \eqref{eq:stat:test:discret} is normalized by an  estimator of $\sigma.$
Proposition 1 states that $\hat\sigma_n/\sigma$ converges in probability to a constant.  
Therefore, it is sufficient to study the asymptotic behaviour of the test statistic as if $\sigma$ was known. Then, in this subsection, we consider the test statistic $T_n$ as : 
\begin{equation}
T_n=\frac{\max_{i=1,\dots,n}\left\| X_{t_i}-X_{t_0} \right\|_2}{\sigma\sqrt{t_n-t_0}}.
\label{app:stat_test}
\end{equation} 

\begin{proof}[\BRD ($H_2$)]
The process $(X_t)$ is a \BRD \eqref{eq:sde_dir} and may be rewritten as,
$$
X_{t_n}-X_{t_0}=v(t_n-t_0)+\sigma(B_{t_n}-B_{t_0}).
$$
Using that $(B_t)$ is a Brownian motion, the distribution of $B_{t_n}-B_{t_0}$ is $\mathcal{N}(\mathbf{0}_2,(t_n-t_0)\mathbf{I}_2).$ Then we have :
\begin{equation}
\mathbb{E}\left(\norme{\frac{X_{t_n}-X_{t_0}}{\sigma(t_n-t_0)}-\frac{v}{\sigma}}^2\right)=\frac{2}{t_n-t_0}.
\label{eq:p_asymp_dir_3}
\end{equation}
As $t_n-t_0=n\Delta,$ we deduce that $V_n=({X_{t_n}-X_{t_0}})/({\sigma(t_n-t_0)})$ converges in probability to $v/\sigma.$  As the euclidean norm is a continuous function, the variable $\|V_n\|$ converges in probability to $\|{v}\|/{\sigma} >0.$
Then $\sqrt{n\Delta}V_n$ converges in probability to $+\infty.$
Since $T_n$ is lower bounded by $\sqrt{n\Delta}V_n=\|({X_{t_n}-X_{t_0}})\|/({\sigma\sqrt{t_n-t_0}})$, the proof is complete.
\end{proof}

\begin{proof}[The Ornstein-Uhlenbeck process ($H_1$)]
The process $(X_t)$ is an Orn\-stein-Uhlenbeck process \eqref{eq:sde_OU}.
We assume that the process is in its stationary regime, that means $X_{t_0}$ is drawn from the stationary distribution that is $X_{t_0}\sim \mathcal{N}(\theta, \sigma^2/(2\lambda)\mathbf{I}_2)$. 
The SDE \eqref{eq:sde_OU} admits an unique solution \citep[Section 2.2.3]{bressloff2014stochastic}
\begin{equation}
X_t-\theta=(X_{t_0}-\theta)e^{-\lambda (t-t_0)}+\sigma  \int\limits_{t_0}^t e^{-\lambda(t- u)} dB^{1/2}_u.
\label{app:eq:ou_sde_int2}
\end{equation}

Then we may bound the test statistic $T_n$ by,
$$
\frac{\|X_{t_0}-\theta\|}{\sigma\sqrt{n\Delta}} + \sum_{i=1}^2 \max_{k=1\ldots n} \frac{|X_{t_k}^i-\theta_i|}{\sigma\sqrt{n\Delta}}.
$$
Since $X_{t_0}$ is drawn from the stationary distribution, the term $\|X_{t_0}-\theta\|/{\sqrt{n\Delta}}$ converges in probability to zero.\newline
Now we show that the second term in the previous equation tends to zero in probability as well.
We introduce the variables $(\xi_k^1,\xi_k^2)$ defined as,
$$
\xi_k^i= (X_{t_k}^i-\theta_i)\sqrt{2\lambda}/\sigma, \qquad k=1\ldots n,\ i=1,2.
$$
Then for $i=1,2,$ the sequence $(\xi_k^i)_k$ is a standardized stationary normal sequence with covariance function,
$$
r_k = \mathbb{E}\left(\xi_\ell^i\xi_{\ell+k}^i\right)=e^{-k\Delta},\quad k\geq-\ell.
$$
Let $i$ be in $\{1,2\}.$
Then $(a_n(\max_{k=1\ldots n}(\xi_k^i)-b_n))_n$ converges in distribution according to \cite[Theorem 4.3.3]{leadbetter1983extremes}, where $a_n=\sqrt{2\log(n)}$ and $b_n=a_n-(2a_n)^{-1}(\log\log(n)+\log(4\pi)).$
We deduce that $\max_{k=1\ldots n}(\xi_k^i)/\sqrt{n\Delta}$ converges in probability to $0.$
Moreover, since $(\xi_k^i)_k$ is a centred Gaussian process, then $\max_{k=1\ldots n}(-\xi_k^i)/\sqrt{n\Delta}$ converges in probability to $0$ by symmetry.
Then we conclude that $\max_{k=1\ldots n} |X_{t_k}^i-\theta_i|/{\sqrt{n\Delta}}$ converges in probability to $0.$
\end{proof}

\begin{proof}[The fractional Brownian Motion ($H_1$)]
The process $(X_t)$ is a fractional Brownian motion with $\mathfrak{h}\in(0,1/2).$ 
From the property of self-similarity and stationarity of increments of the fractional Brownian motion, the following process,
\begin{equation}
Z_t^{(n)}= \frac{X_{tn\Delta+t_0}-X_{t_0}}{\sigma(n\Delta)^{\mathfrak{h}}},\quad t\in[0,1],
\end{equation}
is a fractional Brownian motion. 
We rewrite the test statistic as,
$$
T_n = \frac{1}{(n\Delta)^{1/2-h}} \max_{k=1\ldots n}\left\|Z_{k/n}^{(n)} \right\|
$$
Then $T_n$ is bounded by,
$$
\frac{1}{(n\Delta)^{1/2-h}} \sum_{i=1}^2 \max_{k=1\ldots n} \left|Z_{k/n}^{i,(n)} \right|,
$$
where $Z_{t}^{(n)}=(Z_{t}^{1,(n)},Z_{t}^{2,(n)}).$
The process $Z^{(n)}$ has a version with continuous path as a result of being $\gamma$-Holder continuous for any $\gamma<\mathfrak{h}.$
Let $i\in\{1,2\}$ be fixed.
Then the random variable $\max_{k=1\ldots n} \left|Z_{k/n}^{i,(n)} \right|$ is bounded by,
$$
M_i^{(n)} = \sup_{t\in[0,1]} \left|Z_{t}^{i,(n)} \right|,
$$
which possesses an absolutely continuous density on $\mathbb{R}_+^*$ according to \citet{zaidi2003smoothness}.
That means the sequence $\left(\max_{k=1\ldots n}\left\|Z_{k/n}^{(n)} \right\|\right)_n$ is tight.
Since $\mathfrak{h}<1/2,$ we deduce that $T_n$ converges in probability to $0.$
\end{proof}

\begin{proof}[The fractional Brownian Motion ($H_2$)]
The process $(X_t)$ is a \fBm with $\mathfrak{h}\in(1/2,1).$ 
From the property of self-similarity we get that:
\begin{equation}
Y_n= \frac{\norme{X_{t_n}-X_{t_0}}_2^2}{\sigma^2(t-t_0)^{2\mathfrak{h}}}\sim \chi^2(2).
\end{equation}
We observe that $T_n^2\geq Y_n (n\Delta)^{2\mathfrak{h}-1}.$
Let $x$ be a positive constant.
We have :
\begin{align}
P\left(T_n<x\right)&\leq P\left(Y_n (n\Delta)^{2\mathfrak{h}-1} <x^2\right)\nonumber\\
&\leq P\left(Y_n < x^2/(n\Delta)^{2\mathfrak{h}-1}\right).\label{eq:proof_fbm2}
\end{align}
Since $\mathfrak{h}>1/2,$ $x^2/(n \Delta)^{2\mathfrak{h}-1}$ converges to $0$ as $n\rightarrow \infty.$ 
Then the right hand side of \eqref{eq:proof_fbm2} converges to $0.$ 
That means $P(T_n<x)$ converges to $0$ as $n\to\infty$ : $T_n$ converges to $+\infty$ in probability.
\end{proof}


\subsection{Dependency of the power on the parameters of the parametric alternatives}
\label{app:proof_power}

\begin{lemma}\label{lem:power:brd}
Let $(X_t)$ be a \BRD \eqref{eq:sde_dir}.
Let $\hat{\sigma}_n$ be the estimator of the diffusion coefficient defined in Equation \eqref{eq:sigma_standard}.
The distribution of $T_n$ \eqref{eq:stat:test:discret} depends only on the parameter $v\sqrt{\Delta}/\sigma$ and the trajectory size $n$.
\end{lemma}

\begin{proof}[Proof of Lemma \ref{lem:power:brd}]
We may rewrite the strong solution of the SDE \eqref{eq:sde_dir} as,
\begin{equation*}
X_{t_{k}}=X_{t_{k-1}}+v \Delta+\sigma \sqrt{\Delta} \epsilon_{k},\qquad k=1\ldots n,
\end{equation*}
where $\sqrt{\Delta}\epsilon_k=B_{t_{k}}-B_{t_{k-1}},$ and $(B_t)$ is a standard Brownian motion.
Then $(\epsilon_k)$ is a sequence of independent Gaussian variables $\mathcal{N}(0,1).$
Furthermore, we have immediately :
\begin{equation*}
X_{t_{k}}-X_{t_0}= vk \Delta+\sigma \sqrt{\Delta}\sum\limits_{i=1}^{k} \epsilon_i,\qquad k=1\ldots n.
\end{equation*}
Finally the test statistic $T_n$ may be rewritten as,
\begin{align*}
T_n=\frac{\max_{k=1,\dots,n}\norme{k\frac{v \sqrt{\Delta}}{\sigma}+\sum\limits_{i=1}^k \epsilon_i}}{\sqrt{\frac{1}{2}\sum\limits_{i=1}^n\norme{\frac{v \sqrt{\Delta}}{\sigma}+\epsilon_i}^2}}.
\end{align*}
As the distribution of $(\epsilon_k)$ is free of the parameters the distribution of $T_n$ depends only on $v\sqrt{\Delta}/\sigma$.
\end{proof}

\begin{lemma}\label{lem:power:fbm}
Let $(X_t)$ be a fractional Brownian motion \eqref{eq:sde_FBM}.
Let $\hat{\sigma}_n$ be the estimator of the diffusion coefficient defined in Equation \eqref{eq:sigma_standard}.
The distribution of $T_n$ \eqref{eq:stat:test:discret} depends only on the parameter $\mathfrak{h}$ and the trajectory size $n$.
\end{lemma}

\begin{proof}[Proof of Lemma \ref{lem:power:fbm}]
The fractional Brownian motion may be described by its incremental process \cite{taqqu2003fractional} :
\begin{equation}
\epsilon_k = (X_{t_{k}}-X_{t_{k-1}})/(\sigma \Delta^\mathfrak{h}),\qquad k\geq 1,
\end{equation}
where $(\epsilon_k)$ is a fractional Gaussian noise which is a stationary standardized Gaussian process with autocovariance function $\mathbb{E}(\epsilon_k \epsilon_{k+i})=(1/2)(|i+1|^{2\mathfrak{h}}-2|i|^{2\mathfrak{h}}+|i-1|^{2\mathfrak{h}})$.
Finally the test statistic $T_n$ may be rewritten as,
\begin{align*}
T_n=\frac{\max_{k=1,\dots,n}\norme{\sum\limits_{i=1}^k \epsilon_i}}{\sqrt{\frac{1}{2}\sum\limits_{i=1}^n\norme{\epsilon_i}^2}}.
\end{align*}
Then the distribution of $T_n$ depends only on the trajectory size $n$ and on $\mathfrak{h}$ through the distribution of $(\epsilon_k).$ 
\end{proof}
\newpage
\section{Algorithm}
\begin{algorithm}
\KwIn{$n$, $\alpha,$ $L$}
\tcp{the length $n$ of the trajectory}
\tcp{the probability $\alpha\in(0,1)$}
\tcp{the number $N$ of Monte Carlo experiments }
\KwResult{$q^{(N)}_n(\alpha)$.}
\For{i=1 \emph{\KwTo} N}{
\tcp{Simulation of a Brownian trajectory of size $n$, of variance $\sigma=1$ and with resolution time $\Delta=1$.}
initialization $Y^{(i)}_0=(0,0)^\top$\;
\For{j=1 \emph{\KwTo} n}{
Draw $\epsilon\sim\mathcal{N}(0\mathbin{,}\mathbf{I_2})$\;
$Y^{(i)}_{j}=Y^{(i)}_{j-1}+\epsilon$\;
}
\tcp{Computation of the test statistic}
Compute the ratio $T^{(i)}_n=D_n^{(i)}/\hat{\sigma}_{n}^{(i)}$ from $(Y_0^{(i)},\dots,Y_n^{(i)})$\;
}

\caption[Generation of a $N$-sample $(T^{(1)}_n,\dots,T^{(N)}_n)$.]{Simulation of a $N$-sample $(T^{(1)}_n,\dots,T^{(N)}_n)$ of the distribution of the statistic $T_n$ under $H_0$.} 
\label{algo:mc:sample}
\end{algorithm}
\newpage
\section{Supplementary figures}\label{app:sec:figures}

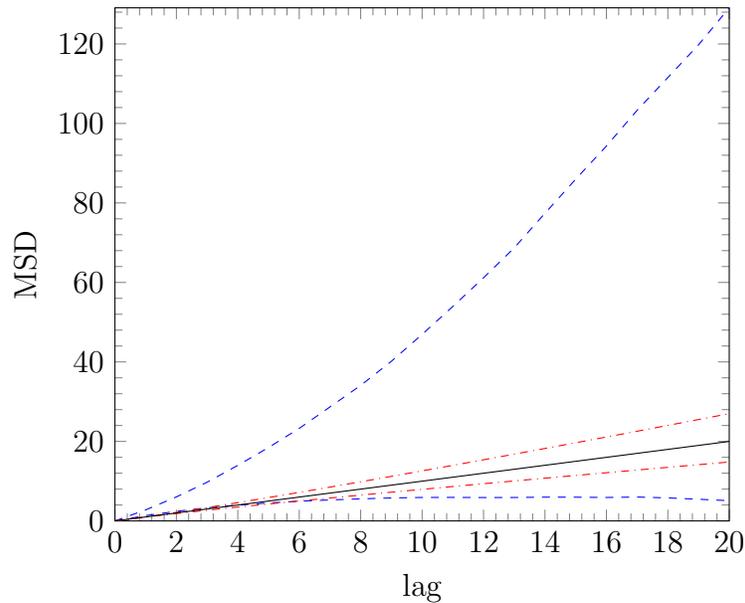
\begin{figure}[h!]
\centering
\pgfplotstableread {fig/msd_curve_cut.txt} {\loadedtable}
\pgfplotsset{
    compat=newest,
    every axis/.append style={
        legend image post style={xscale=0.5}
    }
}
\begin{tikzpicture}[baseline]
\begin{axis}
[xlabel=lag,ylabel=MSD,
minor tick num=4,
enlarge x limits=false,
enlarge y limits=false,
width=0.6\linewidth,
legend style={
   at={(-0.2,-0.1)}, anchor=north west, legend columns=2,font=\tiny}
] 

\addplot[color=blue, dashed] table[x=lag,y=conf_int_inf] from \loadedtable;
\addplot[color=blue, dashed] table[x=lag,y=conf_int_sup] from \loadedtable;
\addplot[color=red, dash dot] table[x=lag,y=feder_inf] from \loadedtable;
\addplot[color=red, dash dot] table[x=lag,y=feder_sup] from \loadedtable;
\addplot[color=black] table[x=lag,y=lag] from \loadedtable;


\end{axis}
\end{tikzpicture}
\caption{A classification rule for motion modes from MSD. The plain line is the theoretical MSD of the standard Brownian motion. The dashdotted lines are the bounds defined by \cite{feder1996constrained}, $t\to t^{\beta},$ $\beta=0.9$ and $1.1.$ If $\hat{\beta}$ the estimation of $\beta$ is such that $0.9<\hat{\beta}<1.1$ it is classified as Brownian motion. The dashed lines are the pointwise high probability interval of 95\% associated to the empirical MSD curve for a standard Brownian motion trajectory of length $n=30.$ 
The bounds of the interval are the 2.5\% and 97.5\% empirical quantile of \eqref{eq:msd_est} and are computed  by Monte Carlo simulation from $10\,001$ Brownian trajectories of size $n=30.$}
\label{fig:msd_curve}
\end{figure}

\begin{table}[!h]
\caption{Monte Carlo estimate of the FDR and mdFDR for both standard and adaptive Procedure \ref{proc1} at level $\alpha=0.05.$ The number of replications is $10\,001.$ 
The error rate estimations are expressed in percentages.}
\centering
\begin{tabular}{cccccc}
\hline 
& &\multicolumn{2}{c}{Standard}&\multicolumn{2}{c}{Adaptive}\\
m & $m_0/m$ & FDR & mdFDR & FDR & mdFDR \\ 
\hline 
100 &0 & 0 & 0  & 0 & 0.2 \\ 
 
&0.2 & 1 & 1 & 3.7 & 3.7 \\ 

&0.4 & 2.1 & 2.1 & 4.2 & 4.2 \\ 

&0.6 & 3.2 & 3.2 & 4.7 & 4.7 \\ 
 
&0.8 & 4.1 & 4.1 & 4.8 & 4.8 \\ 
\hline 
200&0 & 0 & 0 & 0 & 0.4 \\ 
 
&0.2 & 1 & 1 & 3.4 & 3.4 \\ 

&0.4 & 2.1 & 2.1 & 4 & 4 \\ 

&0.6 & 3.2 & 3.2 & 4.6 & 4.6 \\ 
 
&0.8 & 4 & 4 & 4.7 & 4.7 \\ 
\hline 
\end{tabular} 
\label{tab:fdr_mdfdr_100}
\end{table} 

\begin{figure}[h!]
\begin{tikzpicture}
\begin{axis}[
	xticklabels={0,0, 0.02,0.04,0.06,0.08,0.1,0.12,0.14,0.16,0.18},
	ytick={1,2},
	ymin=0,
	ymax=3,
	yticklabels={$H_1$, $H_2$},
	width=\linewidth,
	height=0.35\linewidth,
]
	\addplot[
		mark=*,
		color=blue,
		boxplot prepared={
			lower whisker=0.0004,lower quartile=0.0022,
			median=0.0067,
			upper quartile=0.0222,upper whisker=0.0617,
		},
	]
	coordinates{ };
	\addplot[
		mark=*,
		color=red,
		boxplot prepared={
			lower whisker=0,lower quartile=0,
			median=0.0012,
			upper quartile=0.0252,upper whisker=0.1748,
		},
	]
	coordinates{ };
	\addplot[color=green] coordinates { (0.0229,0) (0.0229,3)};
	\addplot[color=orange] coordinates { (0.0768,0) (0.0768,3)};
	\addplot[color=black] coordinates { (0.05,0) (0.05,3)};
\end{axis}
\end{tikzpicture}
\caption{Boxplots of the \Pv $p_{30}$ (Equation \eqref{eq:p-value_tot}) under $H_1$ and $H_2$. We simulate a set of trajectories $\mathcal{X}_m$ with $m=100$ and $m_0=20$ according to the simulation scheme described in Section \ref{sec:exp-res}. We plot the boxplot of the \Pvs $p^{(i:m)}_{30}$ corresponding to each true alternative hypothesises $H_1$ and $H_2$.  The green (respectively orange) line is the threshold $h=p^{(k^*)}$ obtained by the first step of Procedure 1  (respectively Procedure 1) .The null hypothesis is rejected if the \Pv is lower than $h$. The black line is the level $\alpha=5\%$.}
\label{app:fig:boxplot}
\end{figure}
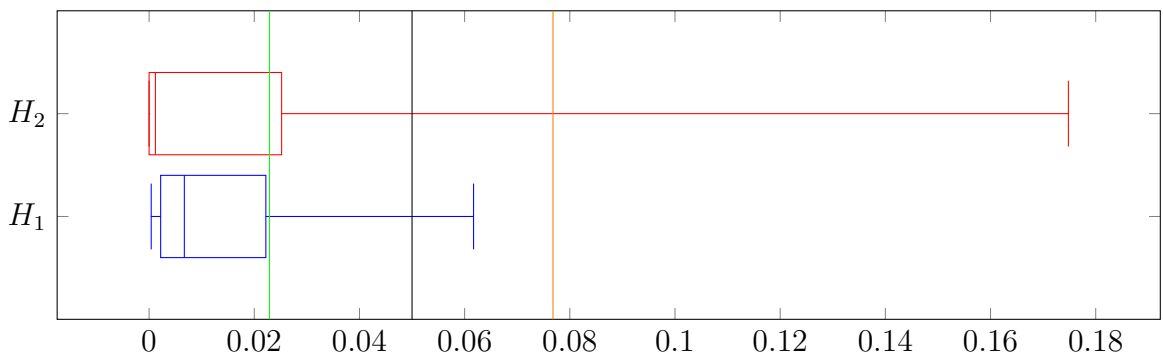

\begin{figure}[h!]
\begin{tikzpicture}
\begin{axis}[
	boxplot/draw direction=y,
	ymin=-0.1,
	ymax=1,
	xtick={1,2,3,4,5,6,7,8,9,10,11,12},
	xticklabels={Br, Sub, Sup,Br, Sub, Sup,Br, Sub, Sup,Br, Sub, Sup},
	width=\linewidth,
	height=0.4\linewidth,
]
	\addplot[
		color=blue,
		boxplot prepared={
			lower whisker=0.3425,lower quartile=0.4266,
			median=0.5301,
			upper quartile=0.5580,upper whisker=0.6386,
		},
	]
	coordinates{};
	\addplot[
		color=blue,
		boxplot prepared={
			lower whisker=0.3012,lower quartile=0.4364,
			median=0.4563,
			upper quartile=0.5690,upper whisker=0.6575,
		},
	]
	coordinates{};
	\addplot[
		color=blue,
		boxplot prepared={
			lower whisker=0,lower quartile=0,
			median=0.0121,
			upper quartile=0.0205,upper whisker=0.0602,
		},
	]
	coordinates{};
	\addplot[
		color=cyan,
		boxplot prepared={
			lower whisker=0.3978,lower quartile=0.5146,
			median=0.6094,
			upper quartile=0.6509,upper whisker=0.7590,
		},
	]
	coordinates{};
	\addplot[
		color=cyan,
		boxplot prepared={
			lower whisker=0.1928,lower quartile=0.3413,
			median=0.3824,
			upper quartile=0.4854,upper whisker=0.6022,
		},
	]
	coordinates{};
	\addplot[
		color=cyan,
		boxplot prepared={
			lower whisker=0,lower quartile=0,
			median=0,
			upper quartile=0.0141,upper whisker=0.0482,
		},
	]
	coordinates{};
		\addplot[
		color=violet,
		boxplot prepared={
			lower whisker=0.2983,lower quartile=0.3998,
			median=0.5365,
			upper quartile=0.5693,upper whisker=0.7048,
		},
	]
	coordinates{};
	\addplot[
		color=violet,
		boxplot prepared={
			lower whisker=0.2470,lower quartile=0.4137,
			median=0.4557,
			upper quartile=0.6002,upper whisker=0.7017,
		},
	]
	coordinates{};
	\addplot[
		color=violet,
		boxplot prepared={
			lower whisker=0,lower quartile=0,
			median=0.0111,
			upper quartile=0.0202,upper whisker=0.0482,
		},
	]
	coordinates{};
			\addplot[
		color=orange,
		boxplot prepared={
			lower whisker=0.0750,lower quartile=0.0863,
			median=0.0967,
			upper quartile=0.1267,upper whisker=0.1627,
		},
	]
	coordinates{};
	\addplot[
		color=orange,
		boxplot prepared={
			lower whisker=0.5904,lower quartile=0.6536,
			median=0.7039,
			upper quartile=0.7326,upper whisker=0.7560,
		},
	]
	coordinates{};
	\addplot[
		color=orange,
		boxplot prepared={
			lower whisker=0.1328,lower quartile=0.1738,
			median=0.2120,
			upper quartile=0.2266,upper whisker=0.2470,
		},
	]
	coordinates{};
\end{axis}
\end{tikzpicture}
\caption{Boxplots of the proportions of Brownian, subdiffusion and superdiffusion computed from 12 Rab11a sequences. In blue proportions obtained with the single test procedure, in cyan with the Procedure 1, in violet with the adaptive Procedure 1 and in orange with the MSD method. Br stands for Brownian, Sub for subdiffusion and Sup for superdiffusion.}
\label{app:fig_rab11}
\end{figure}
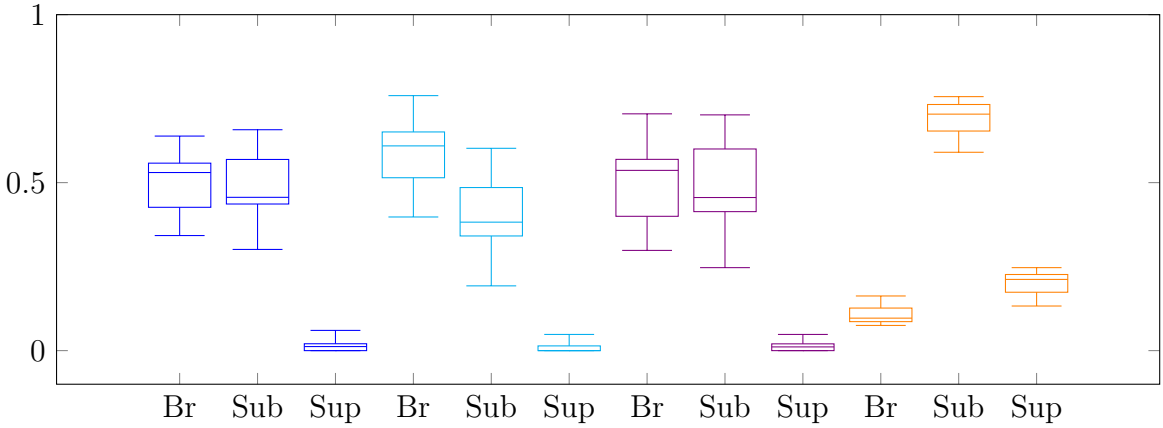

\end{document}